\documentclass[twoside,11pt]{article}
\usepackage[colorlinks = true,
            linkcolor = blue,
            urlcolor  = blue,
            citecolor = blue,
            anchorcolor = blue]{hyperref}
\usepackage{url}
\usepackage{fullpage}
\usepackage[round]{natbib}
\usepackage{caption}

\usepackage{comment}
\usepackage[utf8]{inputenc}
\usepackage{amsmath}
\usepackage{amssymb}

\PassOptionsToPackage{colorlinks = true,
            linkcolor = blue,
            urlcolor  = blue,
            citecolor = blue,
            anchorcolor = blue}{hyperref}
\usepackage{xcolor}
\usepackage{graphicx}
\PassOptionsToPackage{dvipsnames,table}{xcolor}
\usepackage{colortbl}

\usepackage{amsfonts}
\usepackage{multirow,tabularx}
\usepackage{listings}
\usepackage{makecell}
\usepackage[labelfont=bf]{caption} 

\usepackage{algorithm}
\usepackage[noend]{algpseudocode}
\algrenewcommand\algorithmicrequire{\textbf{Input:}}
\algrenewcommand\algorithmicensure{\textbf{Output:}}

\usepackage{amssymb}
\usepackage{graphicx}
\usepackage{amsmath}
\usepackage{amsthm}
\usepackage{xfrac,nicefrac} 
\usepackage{mathrsfs}
\usepackage{subcaption}
\usepackage{xspace}
\usepackage{cleveref}
\usepackage{setspace}

\usepackage[scaled]{helvet} 
\usepackage{courier} 
\normalfont
\usepackage[T1]{fontenc}

\usepackage{changepage}
\usepackage{enumitem}
\usepackage{array}
\usepackage{booktabs}
\usepackage{bbm}
\usepackage{soul}
\usepackage{relsize}
\usepackage{eso-pic}
\usepackage{epsfig}
\usepackage{rotating}
\usepackage{balance}
\usepackage{wrapfig}
\usepackage{framed}
\usepackage{mathtools}
\usepackage{pifont}
\usepackage{scrextend}
\usepackage{sidecap}
\usepackage{adjustbox}
\usepackage{csquotes}
\usepackage{minitoc}

\usepackage{etoolbox}

\usepackage{tikz,pgfplots}
\usepackage{pgfplotstable}
\usetikzlibrary{shapes}
\usetikzlibrary{positioning}
\usetikzlibrary{plotmarks}
\usetikzlibrary{patterns}
\usetikzlibrary{intersections,shapes.arrows}
\usetikzlibrary{arrows.meta}
\usetikzlibrary{pgfplots.fillbetween}
\usepgfplotslibrary{groupplots}
\usetikzlibrary{fit}

\usepackage{tcolorbox}

\definecolor{azure}{rgb}{0.0, 0.5, 1.0}
\tcbuselibrary{skins}
\tcbset{enhanced}

\newcommand{\Ptask}{\mathcal{D}_{{\sf task}}}
\newcommand{\Puser}{\mathcal{D}}
\newcommand{\pRef}{p_{{\sf ref}}}

\newcommand{\ft}{{{\sf FT}}}

\newcommand{\myparagraph}[1]{\noindent{\textbf{#1}.}}  

\newcommand \Pcal {\mathcal{P}}

\usepackage{bm}

\newcommand \xv {{\bm{x}}}

\newcommand \prob {\mathbb{P}}
\newcommand \expect {\mathbb{E}}
\newcommand{\kl}{{\mathrm{KL}}}

\newcommand \pow [1]{^{(#1)}}
\DeclarePairedDelimiterX{\inp}[2]{\langle}{\rangle}{#1, #2} 
\DeclarePairedDelimiterX{\norm}[1]{\Vert}{\Vert}{#1} 
\DeclarePairedDelimiterX{\normsq}[1]{\Vert}{\Vert^2}{#1} 

\renewcommand \epsilon \varepsilon
\newcommand{\eps}{\ensuremath{\epsilon}}

\newtheorem{theorem}{Theorem}

\newtheorem{proposition}[theorem]{Proposition}

\theoremstyle{definition}

\definecolor{C0}{HTML}{1F77B4}
\definecolor{C1}{HTML}{ff7f0e}
\definecolor{C2}{HTML}{2ca02c}
\definecolor{C3}{HTML}{d62728}
\definecolor{C4}{HTML}{9467bd}

\newcommand{\tabemph}[1]{\cellcolor{C1!10}\textcolor{black!90}{#1}}%

\renewcommand{\cite}{\citep}

\newcommand{\krev}[1]{#1}

\newcounter{arxiv}
\renewcommand{\myparagraph}[1]{\paragraph{#1.}\hspace{-0.8em}}  

\title{\textbf{User Inference Attacks on Large Language Models}}

\author{
\begin{tabular}{ccc}
Nikhil Kandpal$^{1}$
&
Krishna Pillutla$^{2}$
&
Alina Oprea$^{2,3}$
\\
Peter Kairouz$^{2}$
&
Christopher A. Choquette-Choo$^{2}$
&
Zheng Xu$^{2}$
\end{tabular}
\vspace{0.5em}
\\
{\small $^1$University of Toronto \& Vector Institute $\qquad$ $^2$Google $\qquad$ $^3$Northeastern University }
}

\date{\vspace{-2em}}

\begin{document}

\maketitle
\doparttoc 
\faketableofcontents 

\begin{abstract}
Fine-tuning is a common and effective method for tailoring large language models (LLMs) to specialized tasks and applications. 
In this paper, we study the privacy implications of fine-tuning LLMs on user data. To this end, we 
\krev{consider}
a realistic threat model, called \emph{user inference}, wherein an attacker infers whether or not a \emph{user's} data was used for fine-tuning. 
We design attacks for performing user inference that require only black-box access to the fine-tuned LLM and a few samples from a user which need not be from the fine-tuning dataset. We find that LLMs are susceptible to user inference across a variety of fine-tuning datasets, at times with near perfect attack success rates. 
Further, we theoretically and empirically investigate the properties that make users vulnerable to user inference, finding that outlier users, users with identifiable shared features between examples, and users that contribute a large fraction of the fine-tuning data are most susceptible to attack.
Based on these findings, we identify several methods for mitigating user inference  including training with example-level differential privacy, removing within-user duplicate examples, and reducing a user's contribution to the training data. While these techniques provide partial mitigation of user inference, we highlight the need to develop methods to fully protect fine-tuned LLMs against this privacy risk.  
\end{abstract}

\section{Introduction}
Successfully applying large language models (LLMs) to real-world problems is often best achieved by fine-tuning on domain-specific data \cite{liu2022fewshot,Mosbach2023FewshotFV}. This approach is seen in a variety of commercial products deployed today, e.g., GitHub Copilot~\cite{chen2021evaluating}, Gmail Smart Compose~\cite{dai2019gmail}, GBoard~\cite{xu2023federated}, etc., that are based on LLMs trained or fine-tuned on domain-specific data collected from users. The practice of fine-tuning on user data---particularly on sensitive data like emails, texts, or source code---comes with privacy concerns, as LLMs have been shown to leak information from their training data~\cite{Carlini_Memorization}, especially as models are scaled larger \cite{carlini2023quantifying}. In this paper, we study the privacy risks posed to users whose data are leveraged to fine-tune LLMs.


Most existing privacy attacks on LLMs can be grouped into two categories: \emph{membership inference}, in which the attacker obtains access to a sample and must determine if it was trained on~\cite{mireshghallah-etal-2022-quantifying,mattern-etal-2023-membership,code_leakage}; and \emph{extraction attacks}, in which the attacker tries to reconstruct the training data by prompting the model with different prefixes~\cite{Carlini_Memorization,pii_leakage}.
These threat models make no assumptions about the origin of the training data and thus cannot estimate the privacy risk to a user that contributes many training samples that share characteristics (e.g., topic, writing style, etc.). 
To this end, we consider the threat model of user inference~\cite{miao2021audio,Hartmann_DistrInf_2023} for the first time for LLMs. \emph{We show that user inference is a realistic privacy attack for LLMs fine-tuned on user data.}

\begin{figure*}
\centering
\includegraphics[width=0.95\linewidth]{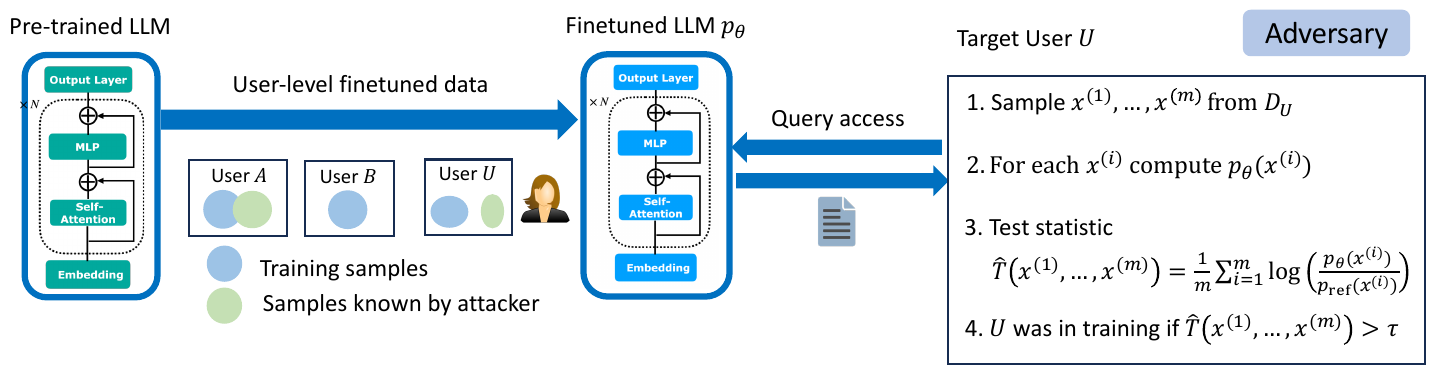}
\caption{The user inference threat model. An LLM is fine-tuned on user-stratified data. The  adversary can query samples on the fine-tuned model to compute likelihoods. The adversary can access samples from a user's distribution (different than the user training samples) to compute a likelihood score to determine if the user participated in training.} \label{fig:attack_overview}
\end{figure*}

In user inference (see \Cref{fig:attack_overview}), the attacker aims to determine if a particular user participated in LLM fine-tuning using only a few \emph{fresh} samples from the user and black-box access to the fine-tuned model.
This threat model lifts membership inference from the privacy of individual samples to the privacy of users who contribute multiple samples, while also relaxing the stringent assumption that the attacker has access to the exact fine-tuning data.
By itself, user inference could be a privacy threat if the fine-tuning task reveals sensitive information about participating users (e.g., a model is fine-tuned only on users with a rare disease). Moreover, user inference may also enable other attacks extracting sensitive information about specific users, similar to how membership inference is used as a subroutine in training data extraction attacks~\cite{Carlini_Memorization}. 

In this work, we construct a simple and practical user inference attack that determines if a user participated in LLM fine-tuning. It involves computing a likelihood ratio test statistic normalized relative to a reference model (\Cref{sec:attack}). This attack can be efficiently mounted even at the LLM scale. We empirically study its effectiveness on the GPT-Neo family of LLMs~\cite{gpt-neo} when fine-tuned on diverse data domains, including emails, social media comments, and news articles (\Cref{sec:results}). 
This study gives insight into the various parameters that affect vulnerability to user inference---such as uniqueness of a user's data distribution, amount of fine-tuning data contributed by a user, and amount of attacker knowledge about a user.

We evaluate the attack on synthetically generated canary users to characterize the privacy leakage for worst-case users (\Cref{sec:canaries}).
We show that canary users constructed via minimal modifications to the real users' data increase the attack's effectiveness (in AUROC) by up to $40\%$. This indicates that simple features shared across a user's samples like an email signature or a characteristic phrase, can greatly exacerbate the risk of user inference.


Finally, we evaluate several methods for mitigating user inference, such as limiting the number of fine-tuning samples contributed by each user, removing duplicates within a user's samples, early stopping, gradient clipping, and fine-tuning with example-level differential privacy (DP). Our results show that duplicates within a user's examples can exacerbate the risk of user inference, but are not necessary for a successful attack. Additionally, limiting a user's contribution to the fine-tuning set can be effective but is only feasible for data-rich applications with a large number of users. Finally, example-level DP provides some defense but is ultimately designed to protect the privacy of individual examples, rather than users that contribute multiple examples. These results highlight the importance of future work on scalable \emph{user-level} DP algorithms that have the potential to provably mitigate user inference~\cite{mcmahan2017learning,levy2021learning}. Overall, we are the first to study user inference against LLMs and provide key insights to inform future deployments of LLMs fine-tuned on user data. 
\section{Related Work}

There are many different ML privacy attacks with different objectives~\cite{NIST_report}: \emph{membership inference} attacks determine if a particular data sample was part of a model's training set~\cite{shokri2017membership,yeom_mi,LiRA,enhanced_mi,watson2022importance,choquette2021label,jagielski2023students}; \emph{data reconstruction} aims to exactly reconstruct the training data of a model, typically for a discriminative model~\cite{haim2022reconstructing}; and \emph{data extraction} attacks aim to extract training data from generative models like LLMs~\cite{Carlini_Memorization,pii_leakage,ippolito2022preventing,anil2023palm,kudugunta2023madlad,nasr2023scalable}.

\myparagraph{Membership inference attacks on LLMs}
\citet{mireshghallah-etal-2022-quantifying} introduce a likelihood ratio-based attack on LLMs, designed for masked language models, such as BERT. \citet{mattern-etal-2023-membership} compare the likelihood of a sample against the average likelihood of a set of neighboring samples, and eliminate the assumption of attacker knowledge of the training distribution used in prior works.
\citet{debenedetti2023privacy} study how systems built on LLMs may amplify membership inference.
\citet{Carlini_Memorization} use a perplexity-based membership inference attack to extract training data from GPT-2. Their attack prompts the LLM to generate sequences of text, and then uses membership inference to identify sequences copied from the training set. 
Note that membership inference requires access to exact training samples while user inference does not.

\myparagraph{Extraction attacks}
Following~\citet{Carlini_Memorization}, memorization in LLMs received much attention~\cite{zhang2021counterfactual,NEURIPS2022_fa0509f4,biderman2023emergent,anil2023palm}. These works found that memorization scales with model size~\cite{carlini2023quantifying} and data repetition \cite{kandpal2022deduplicating}, may eventually be forgotten~\cite{jagielski2022measuring}, and can exist even on models trained for specific restricted use-cases like translation~\cite{kudugunta2023madlad}. 
\citet{pii_leakage} develop techniques to extract PII information from LLMs and  \cite{inan2021training} design metrics to measure how much of user's confidential data is leaked by the LLM.
Once a user's participation is identified by user inference, these techniques can be used to estimate the amount of privacy leakage.

\myparagraph{User-level membership inference}
Much prior work on inferring a user's participation in training makes the stronger assumption that the attacker has access to a user's exact training samples. We call this \emph{user-level membership inference} to distinguish it from \emph{user inference} (which does not require access to the exact training samples).
\citet{song2019auditing} give the first such an attack for generative text models.
Their attack is based on training multiple shadow models and does not scale to LLMs. This threat model has also been studied for text classification via reduction to membership inference~\citep{shejwalkar2021membership}.

\myparagraph{User inference} This threat model was considered for
speech recognition in IoT devices~\cite{miao2021audio}, representation learning~\cite{li2022userlevel} and face recognition~\cite{face_auditor}. 
\citet{Hartmann_DistrInf_2023} formally define user inference for classification and regression but call it \emph{distributional membership inference}.
These attacks are domain-specific or require shadow models. Thus, they do not apply or scale to LLMs. Instead, we design an efficient user inference attack that scales to LLMs and illustrate the user-level privacy risks posed by fine-tuning on user data. See \Cref{sec:a:related} for further discussion.


\section{User Inference Attacks} \label{sec:attack}
Consider an autoregressive language model $p_\theta$ that defines a distribution $p_\theta(x_t | \xv_{<t})$ over the next token $x_t$ in continuation of a prefix $\xv_{<t} \doteq (x_1, \ldots, x_{t-1})$.
We are interested in a setting where a pretrained LLM $p_{\theta_0}$ with initial parameters $\theta_0$ is fine-tuned on a dataset $D_\ft$ sampled i.i.d. from a distribution $\Ptask$. 
The most common objective is to minimize the cross entropy of predicting each next token $x_t$ given the context $\xv_{<t}$ for each fine-tuning sample $\xv \in D_\ft$.
Thus, the fine-tuned model $p_\theta$ is trained to maximize the log-likelihood $\sum_{\xv \in D_\ft} \log p_\theta(\xv) = \sum_{\xv \in D_\ft} \sum_{t=1}^{|\xv|} \log p_\theta(x_t | \xv_{<t})$ of the fine-tuning set $D_\ft$.

\myparagraph{Fine-tuning with user-stratified data}
Much of the data used to fine-tune LLMs has a user-level structure. For example, emails, messages, and blog posts can reflect the specific characteristics of their author. Two text samples from the same user are more likely to be similar to each other than samples across users in terms of language use, vocabulary, context, and topics.
To capture user-stratification, we model the fine-tuning distribution $\Ptask$ as a mixture
\begin{align} \label{eq:task-distr}
\textstyle
   \Ptask = \sum_{u=1}^n \alpha_u \Puser_u
\end{align}
of $n$ user data distributions $\Puser_1, \ldots, \Puser_n$ with non-negative weights $\alpha_1,\ldots, \alpha_n$ that sum to one. 
One can sample from $\Ptask$ by first sampling a user $u$ with probability $\alpha_u$ and then sampling a document $\xv \sim \Puser_u$ from the user's data distribution. We note that the fine-tuning process of the LLM is oblivious to user-stratification of the data.

\myparagraph{The user inference threat model}
The task of membership inference assumes that an attacker has access to a text sample $\xv$ and must determine whether that particular sample was a part of the training or fine-tuning data~\cite{shokri2017membership,yeom_mi,LiRA}. The \textbf{user inference} threat model relaxes the assumption that the attacker has access to samples from the fine-tuning data.

The attacker aims to determine if \emph{any} data from user $u$ was involved in fine-tuning the model $p_\theta$ using $m$ i.i.d. samples $\xv\pow{1:m} := (\xv\pow{1}, \ldots, \xv\pow{m}) \sim \Puser_u^m$ from user $u$'s distribution. Crucially, we allow $\xv\pow{i} \notin D_\ft$, i.e., the attacker is not assumed to have access to the exact samples of user $u$ that were a part of the fine-tuning set. For instance, if an LLM is fine-tuned on user emails, the attacker can reasonably be assumed to have access to \emph{some} emails from a user, but not necessarily the ones used to fine-tune the model. We believe this is a realistic threat model for LLMs, as it does not require exact knowledge of training set samples, as in membership inference attacks. 

We assume that the attacker has \emph{black-box access} to the LLM $p_\theta$ --- they can only query the model's likelihood on a sequence of tokens and might not have knowledge of either the model architecture or parameters. Following standard practice in membership inference~\cite{mireshghallah-etal-2022-quantifying,watson2022importance}, we allow the attacker access to a reference model $\pRef$ that is similar to the target model $p_\theta$ but has not been trained on user $u$'s data.
This can simply be the pre-trained model $p_{\theta_0}$ or another LLM.




\myparagraph{Attack strategy}
The attacker's task can be formulated as a statistical hypothesis test. Letting $\Pcal_u$ denote the set of models trained on user $u$'s data, the attacker aims to test:
\begin{align}
    H_0 \, :\, p_\theta \notin \Pcal_u, \qquad 
    H_1 \, :\, p_\theta \in \Pcal_u \,.
\end{align}
There is generally no prescribed recipe to test for such a composite hypothesis. 
Typical attack strategies involve training multiple ``shadow'' models~\cite{shokri2017membership}; see \Cref{appendix:other_methods}. This, however, is infeasible at LLM scale.

The likelihood under the fine-tuned model $p_\theta$ is a natural test statistic: we might expect $p_\theta(\xv\pow{i})$ to be high if $H_1$ is true and low otherwise. Unfortunately, this is not always true, even for membership inference. Indeed, 
$p_\theta(\xv)$ can be large for $\xv \notin D_\ft$ for easy-to-predict $\xv$ (e.g., generic text using common words), while $p_\theta(\xv)$ can be small even if $\xv \in D_\ft$ for hard-to-predict $\xv$. This necessitates the need for calibrating the test using a reference model~\cite{mireshghallah-etal-2022-quantifying,watson2022importance}.

We overcome this difficulty by replacing the attacker's task with surrogate hypotheses that are easier to test efficiently:
\begin{align}
\begin{aligned}
    H_0'\, &: \, \xv\pow{1:m} \sim \pRef^m \,, \qquad
    H_1' \, :\, \xv\pow{1:m} \sim p_\theta^m \,.
\end{aligned}
\end{align}

%

By construction, $H_0'$ is always false since $\pRef$ is not fine-tuned on user $u$'s data. However, $H_1'$ is more likely to be true if the user $u$ participates in training \emph{and}
the samples contributed by $u$ to the fine-tuning dataset $D_\ft$ are similar to the samples $\xv\pow{1:m}$ known to the attacker even if they are not identical. In this case, the attacker rejects $H_0'$. Conversely, if user $u$ did not participate in fine-tuning and no samples from $D_\ft$ are similar to $\xv\pow{1:m}$, then the attacker finds both $H_0'$ and $H_1'$ to be equally (im)plausible, and fails to reject $H_0'$.
Intuitively, to faithfully test $H_0$ vs. $H_1$ using $H_0'$ vs. $H_1'$, we require that
$\xv, \xv' \sim \Puser_u$ are closer on average than $\xv \sim \Puser_u$ and $\xv'' \sim \Puser_{u'}$ for any other $u' \neq u$.


The Neyman-Pearson lemma tells us that the \emph{likelihood ratio test} is the most powerful for testing $H_0'$ vs. $H_1'$, i.e., it achieves the best true positive rate at any given false positive rate~\cite[Thm. 3.2.1]{lehmann1986testing}. This involves constructing a test statistic using the log-likelihood ratio
\begin{align}
\begin{aligned}
    T(\xv\pow{1}, \ldots, \xv\pow{m})
    &:= \log \left( \frac{p_\theta(\xv\pow{1}, \ldots, \xv\pow{m})}{\pRef(\xv\pow{1}, \ldots, \xv\pow{m})} \right) \\
    &= \sum_{i=1}^m \log \left( \frac{p_\theta(\xv\pow{i})}{\pRef(\xv\pow{i})} \right) \,,
\end{aligned}
\end{align}
where the last equality follows from the independence of each $\xv\pow{i}$, which we assume.
Although independence may be violated in some domains (e.g. email threads), it makes the problem more computationally tractable. As we shall see, this already gives us relatively strong attacks.

Given a threshold $\tau$, the attacker rejects the null hypothesis and declares that $u$ has participated in fine-tuning if $T(\xv\pow{1}, \ldots, \xv\pow{m}) > \tau$.
In practice, the number of samples $m$ available to the attacker might vary for each user, so we normalize the statistic by $m$. Thus, our final attack statistic is $\hat T(\xv\pow{1}, \dots, \xv\pow{m}) = \tfrac{1}{m} \, T(\xv\pow{1}, \dots, \xv\pow{m})$.


\begin{table*}[t!]
\centering
\renewcommand{\arraystretch}{1.2}
\small
\begin{tabular}{llrrrrrrr} 
\toprule
\multirow{2}{*}{\textbf{Dataset}} & 
\multirow{2}{*}{\textbf{User Field}} & 
\multirow{2}{*}{\textbf{\#Users}} &
\multirow{2}{*}{\textbf{\#Examples}} &
\multicolumn{5}{c}{\textbf{Percentiles of Examples/User}} \\
\cmidrule{5-9}
& & & & \multicolumn{1}{l}{$\mathbf{P_{0}}$} & \multicolumn{1}{l}{$\mathbf{P_{25}}$} & \multicolumn{1}{l}{$\mathbf{P_{50}}$} & \multicolumn{1}{l}{$\mathbf{P_{75}}$} & \multicolumn{1}{l}{$\mathbf{P_{100}}$}
\\ \hline
Reddit Comments    & 
User Name              &
$5194$      &
$1002K$  &  
$100$      &
$116$      &
$144$      &
$199$      &
$1921$
\\
CC News    & 
Domain Name                       & 
$2839$    &
$660K$        &  
$30$      &
$50$      &
$87$      &
$192$     &
$24480$
\\
Enron Emails        & 
Sender's Email Address                        & 
$136$  &
$91K$    &  
$28$   &
$107$   &
$279$  &
$604$  &
$4280$
\\ \bottomrule
\end{tabular}
\caption{\small \textbf{Evaluation dataset summary statistics}: The three evaluation datasets vary in their notion of ``user'' (i.e. a Reddit comment belongs to the username that it was posted from whereas a CC News article belongs to the web domain where the article was published). Additionally, these datasets span multiple orders of magnitude in terms of number of users and number of examples contributed per user.} \label{tab:datasets}
\end{table*}

\myparagraph{Analysis of the attack statistic}
We analyze this attack statistic in a simplified setting to gain some intuition. 
In the large sample limit as $m \to \infty$, the mean statistic $\hat T$ approximates the population average
\begin{align}
    \bar T(\Puser_{u}) &:= 
    \mathbb{E}_{\xv \sim \Puser_{u}} \left[ \log \left( \frac{p_\theta(\xv)}{\pRef(\xv)} \right) \right] \,.
\end{align}

We will analyze this test statistic for the choice $\pRef = \Puser_{-u} \propto \sum_{u' \neq u} \alpha_{u'} \Puser_{u'}$, which is the fine-tuning mixture distribution excluding the data of user $u$. This is motivated by the results of \citet{watson2022importance} and \citet{sablayrolles2019whitebox}, who show that using a reference model trained on the whole dataset excluding a single sample approximates the optimal membership inference classifier.
Let $\kl(\cdot \Vert \cdot)$ and $\chi^2(\cdot \Vert \cdot)$ denote the Kullback–Leibler and $\chi^2$ divergences. We establish a bound (proved in \Cref{appendix:theory}) assuming $p_\theta,\pRef$ perfectly capture their target distributions.

\begin{proposition}
\label{thm:analysis}
Assume $p_\theta = \Ptask$ and $\pRef = \Puser_{-u}$ for some user $u \in [n]$. Then, we have
\[ 
    \log \left( \alpha_{u} \right) + \kl(\Puser_{u} \parallel \Puser_{-u})
    <
    \bar T(\Puser_{u})
    \le 
    \alpha_{u} \, \chi^2(\Puser_{u} \Vert \Puser_{-u})
    \,.
\]
\end{proposition}

This suggests the attacker may more easily infer:
\begin{enumerate}[label=(\alph*),nosep, leftmargin=1.5em]
    \item users who contribute more data (so $\alpha_{u}$ is large), or
    \item users who contribute unique data (so $\kl(\Puser_{u} \Vert \Puser_{-u})$ and $\chi^2(\Puser_{u} \Vert \Puser_{-u})$ are large).
\end{enumerate}
Conversely, if neither holds, then a user's participation in fine-tuning cannot be reliably detected. Our experiments corroborate these and we use them to design mitigations.

\section{Experiments}
In this section, we empirically study the susceptibility of models to user inference attacks, the factors that affect attack performance, and potential mitigation strategies.

\begin{figure*}
\centering
\includegraphics[width=0.95\linewidth]{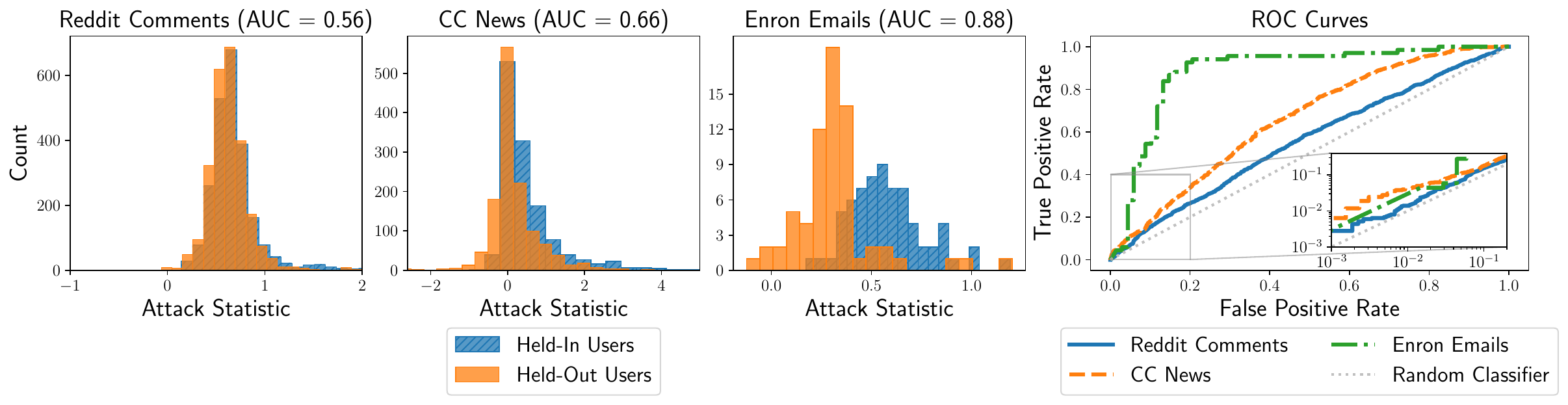}
\caption{\small \textbf{Our attack can achieve significant AUROC}, e.g., on the Enron emails dataset.
\textbf{Left three}: Histograms of the test statistics for held-in and held-out users for the three attack evaluation datasets. \textbf{Rightmost}: Their corresponding ROC curves. 
}
\label{fig:hists_aurocs}
\end{figure*}

\subsection{Experimental Setup}

\myparagraph{Datasets}
We evaluate user inference attacks on three user-stratified text datasets representing different  domains: Reddit Comments \cite{pushshift_reddit} for social media content, CC News\footnote{
    While CC News does not strictly have user data, it is made up of non-identical groups (as in Eq. \eqref{eq:task-distr}) defined by the web domain. We treat each group as a ``user'' as in~\citet{charles2023towards}.
} \cite{Hamborg2017} for news articles, and Enron Emails \cite{Klimt2004IntroducingTE} for user emails.
These datasets are diverse in their domain, notion of a user, number of users, and amount of data contributed per user (\Cref{tab:datasets}).
We also report results for the ArXiv Abstracts dataset \cite{clement2019arxiv} in \Cref{appendix:expt-results}.

To make these datasets suitable for evaluating user inference, we split them into a held-in set of users to fine-tune models, and a held-out set of users to evaluate attacks. Additionally, we set aside 10\% of each user's samples as the samples used by the attacker to run user inference attacks; these samples are not used for fine-tuning.
For more details on the dataset preprocessing, see \Cref{appendix:experiments}.

\myparagraph{Models}
We evaluate user inference attacks on the $125$M and $1.3$B parameter decoder-only LLMs from the GPT-Neo \citep{gpt-neo} model suite. These models were pre-trained on The Pile dataset \citep{gao2020pile}, an $825$ GB diverse text corpus, and use the same architecture and pre-training objectives as the GPT-2 and GPT-3 models. Further details on the fine-tuning are given in \Cref{appendix:experiments}.

\myparagraph{Attack and Evaluation}
We implement the user inference attack of \Cref{sec:attack} using the pre-trained GPT-Neo models as the reference $\pRef$.
Following the membership inference literature,
we evaluate the aggregate attack success using the Receiver Operating Characteristic (ROC) curve across held-in and held-out users; this is a plot of the true positive rate (TPR) and false positive rate (FPR) of the attack across all possible thresholds. We use the area under this curve (AUROC) as a scalar summary.
We also report the TPR at small FPR (e.g., $1\%$) \citep{LiRA}.

\myparagraph{Remarks on Fine-Tuning Data}
Due to the size of pre-training datasets like The Pile, we found it challenging to find user-stratified datasets that were not part of pre-training; this is a problem with LLM evaluations in general~\cite{sainz}. 
However, we believe that our setup still faithfully evaluates the fine-tuning setting for two main reasons. First, the overlapping fine-tuning data constitutes only a small fraction of all the data in The Pile. Second, our attacks are likely only weakened (and thus, underestimate the true risk) by this setup. This is because inclusion of the held-out users in pre-training should only reduce the model's loss on these samples, making the loss difference smaller and thus our attack harder to employ.
\subsection{User Inference: Results and Properties} \label{sec:results}

We examine how user inference is impacted by factors such as the amount of user data and attacker knowledge, the model scale, as well as the connection to overfitting.


\myparagraph{Attack Performance}
We attack GPT-Neo $125$M fine-tuned on each of the three fine-tuning datasets and evaluate the attack performance. 
We see from \Cref{fig:hists_aurocs} that the user inference attacks on all three datasets achieve non-trivial performance, with the attack AUROC varying between $88\%$ (Enron) to $66\%$ (CC News) and $56\%$ (Reddit).

\begin{figure*}[bt]
\centering
\includegraphics[width=\linewidth]{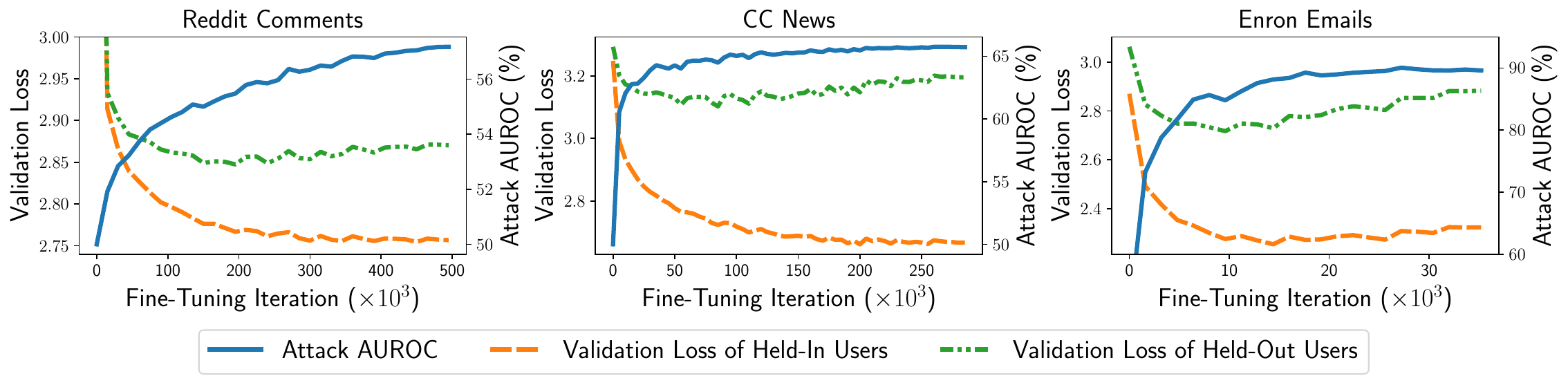}
\caption{\small
\textbf{Attack success over fine-tuning}: User inference AUROC and the held-in/held-out validation loss.
} \label{fig:training_run}
\end{figure*}
\begin{figure*}[t!]
\centering
\includegraphics[width=\linewidth]{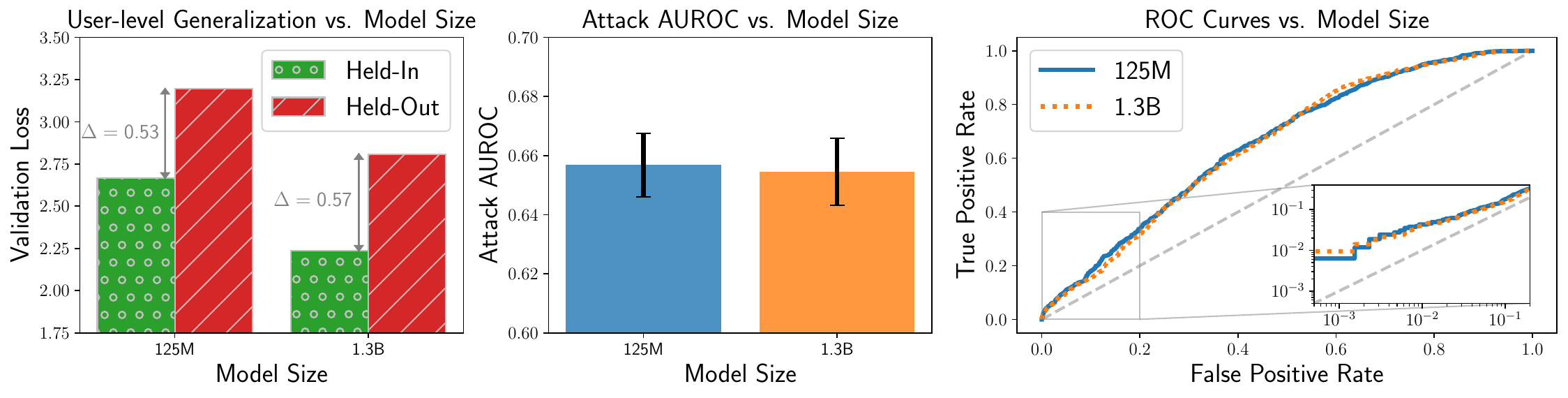}
\caption{\small
\textbf{Attack success vs. model scale}: User inference attack performance in $125$M and $1.3$B models trained on CC News. \textbf{Left}: Although the $1.3$B model achieves lower validation loss, the difference in validation loss between held-in and held-out users is the same as that of the $125$M model. \textbf{Center \& Right}: User inference attacks against the $125$M and $1.3$B models achieve the same performance.} \label{fig:model_scale}
\end{figure*}

\begin{figure}[bth!]
\centering
\includegraphics[width=0.9\linewidth]{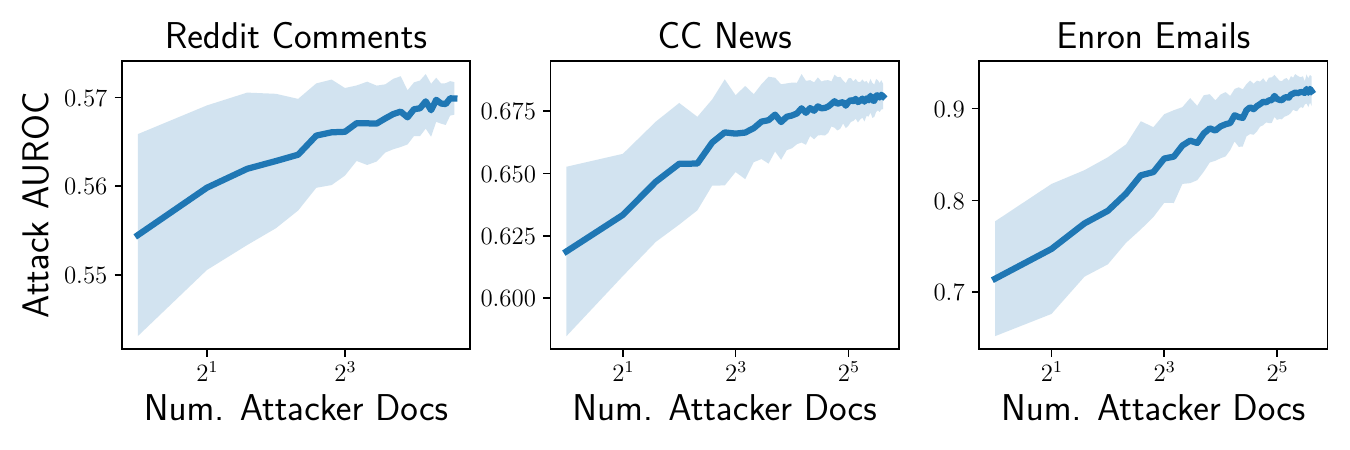}
\caption{\small
\textbf{Attack performance vs. attacker knowledge}: As we increase the number of examples given to the attacker, the attack performance increases across all three datasets. 
The shaded area denotes the std over $100$ random draws of attacker examples.
} \label{fig:attacker_knowledge}
\end{figure}

The disparity in performance between the three datasets can be explained in part by the intuition from \Cref{thm:analysis}, which points out two factors.
First, a larger fraction of data contributed by a user makes user inference easier. The Enron dataset has fewer users, each of whom contributes a significant fraction of the fine-tuning data (cf.~\Cref{tab:datasets}), 
while, the Reddit dataset has a large number of users, each with few datapoints.
Second, distinct user data makes user inference easier. Emails are more distinct due to identifying information such as names (in salutations and signatures) and addresses, while news articles or social media comments from a particular user may share more subtle features like topic or writing style.

\myparagraph{The Effect of the Attacker Knowledge}
We examine the effect of the attacker knowledge
(the amount of user data used by the attacker to compute the test statistic) in \Cref{fig:attacker_knowledge}. First, we find that  more attacker knowledge leads to higher attack AUROC and lower variance in the attack success. For CC News, the AUROC increases from $62.0 \pm 3.3 \%$ when the attacker has only one document to $68.1\pm 0.6\%$ at 50 documents. The user inference attack already leads to non-trivial results with an attacker knowledge of \emph{one document per user} for CC News (AUROC $62.0\%$) and Enron Emails (AUROC $73.2\%$).
Overall, the results show that an attacker does not need much data to mount a strong attack, and more data only helps.

\myparagraph{User Inference and User-level Overfitting}
It is well-established that overfitting to the training data is sufficient for successful membership inference~\cite{yeom_mi}. We find that a similar phenomenon holds for user inference, which is enabled by \emph{user-level overfitting}, i.e., the model overfits not to the training samples themselves, but rather the \emph{distributions} of the training users.

We see from \Cref{fig:training_run} that the validation loss of held-in users continues to decrease for \emph{all 3 datasets}, while the loss of held-out users increases. These curves display a textbook example of overfitting, not to the training data (since both curves are computed using validation data), but to the distributions of the training users. Note that the attack AUROC improves with the widening generalization gap between these two curves.
Indeed, the Spearman correlation between the generalization gap and the attack AUROC is at least $99.4\%$ for all datasets.
This demonstrates the close relation between user-level overfitting and user inference.

\begin{figure*}
\centering
\includegraphics[width=\linewidth]{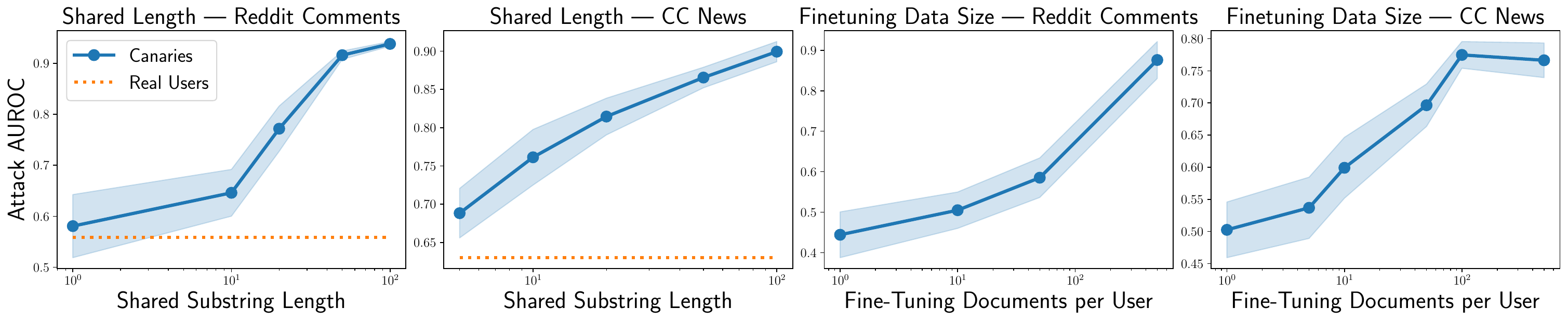}
\caption{\small
Canary experiments. \textbf{Left two}: Comparison of attack performance on the natural distribution of users (``Real Users'') and attack performance on synthetic canary users (each with 100 fine-tuning documents) as the shared substring in a canary's documents varies in length. \textbf{Right two}: Attack performance on canary users (each with a 10-token shared substring) decreases as their contribution to the fine-tuning set decreases.
On all plots, we shade the AUROC std over $100$ bootstrap samples of held-in and held-out users.
} \label{fig:canaries}
\end{figure*}


\myparagraph{Attack Performance and Model Scale}
Next, we investigate the role of model scale in user inference using the GPT-Neo $125$M and $1.3$B on the CC News dataset. 

\Cref{fig:model_scale} shows  that the attack AUROC is nearly identical for the $1.3$B model ($65.3\%$) and $125$M model ($65.8\%$). While the larger model achieves better validation loss on both held-in users ($2.24$ vs. $2.64$) and held-out users ($2.81$ vs. $3.20$), the generalization gap is nearly the same for both models ($0.57$ vs. $0.53$). This shows a qualitative difference between user and membership inference, where attack performance reliably increases with model size in the latter \cite{carlini2023quantifying,NEURIPS2022_fa0509f4,kandpal2022deduplicating,mireshghallah-etal-2022-quantifying,anil2023palm}.

\subsection{User Inference in the Worst-Case} \label{sec:canaries}

The disproportionately large downside to privacy leakage necessitates looking beyond the average-case privacy risk to worst-case settings.
Thus, we analyze attack performance on datasets containing synthetically generated users, known as \emph{canaries}.
There is usually a trade-off between making the canary users realistic and worsening their privacy risk. We intentionally err on the side of making them realistic to illustrate the potential risks of user inference.

To construct a canary user, we first sample a real user from the dataset and insert a particular substring into each of that user's examples. The substring shared between all of the user's examples is a contiguous substring randomly sampled from one of their documents (for more details, see \Cref{appendix:experiments}). We construct $180$ canary users with shared substrings ranging from $1$-$100$ tokens in length and inject these users into the Reddit and CC News datasets. 
We do not experiment with synthetic canaries in Enron Emails, as the attack AUROC already exceeds $88\%$ for real users. 

\Cref{fig:canaries} (left) shows that the attack is more effective on canaries than real users, and increases with the length of the shared substring. 
A short shared substring is enough to significantly increase the attack AUROC from $63\%$ to $69\%$ (5 tokens) for CC News and $56\%$ to $65\%$ for Reddit (10 tokens).

These results raise a question if canary gradients can be filtered out easily (e.g., using the $\ell_2$ norm). However, \Cref{fig:gradient_clipping} (right) shows that the gradient norm distribution of the canary gradients and those of real users are nearly indistinguishable.
This shows that our canaries are close to real users from the model's perspective, and thus hard to filter out. This experiment also demonstrates the increased privacy risk for users who use, for instance, a short and unique signature in emails or characteristic phrases in documents.

\subsection{Mitigation Strategies} \label{sec:defenses}
Finally, we investigate existing techniques for limiting the influence of individual examples or users on model fine-tuning as methods for mitigating user inference attacks.

\myparagraph{Gradient Clipping}
Since we consider fine-tuning that is oblivious to the user-stratification of the data, one can limit the model's sensitivity by clipping the gradients per batch \cite{pascanu2013difficulty} or per example \cite{Abadi_2016}. \Cref{fig:gradient_clipping} (left) plots its effect for the $125$M model on CC News: neither batch nor per-example gradient clipping have any effect on user inference. \Cref{fig:gradient_clipping} (right) tells us why: canary examples do not have large outlying gradients and clipping affects real and canary data similarly. Thus, gradient clipping is an ineffective mitigation strategy.

\myparagraph{Early Stopping}
The connection between user inference and user-level overfitting from \Cref{sec:results} suggests that early stopping, a common heuristic used to prevent overfitting~\cite{caruana2000overfitting}, could potentially mitigate user inference. Unfortunately, we find that $95\%$ of the final AUROC is obtained quite early in training: $15$K steps ($5\%$ of the fine-tuning) for CC News and $90$K steps ($18\%$ of the fine-tuning) for Reddit, see \Cref{fig:training_run}.
Typically, the overall validation loss still decreases far after this point. This suggests an explicit tradeoff between model utility (e.g., in validation loss) and privacy risks from user inference.

\myparagraph{Data Limits Per User}
Since we cannot change the fine-tuning procedure, we consider limiting the amount of fine-tuning data per user.
\Cref{fig:canaries} (right two) show that this can be effective. For CC News, the AUROC for canary users reduces from $77\%$ at $100$ fine-tuning documents per user to almost random chance at $5$ documents per user.
A similar trend also holds for Reddit.

\myparagraph{Data Deduplication}
Since data deduplication can mitigate membership inference~\cite{lee2021deduplicating,kandpal2022deduplicating}, we evaluate it for user inference. CC News is the only dataset in our suite with within-user duplicates (Reddit and Enron are deduplicated in the preprocessing; see \Cref{sec:a:expt:ds}), so we use it for this experiment.\footnote{
    Although each article of CC News from HuggingFace Datasets has a unique URL, the text of $11\%$ of the articles has exact duplicates from the same domain.
    See \S\ref{sec:ccnews-dup} for examples.
}
The deduplication reduces the attack AUROC from $65.7\%$ to $59.1\%$. The attack ROC curve of the deduplicated version is also uniformly lower, even at extremely small FPRs (\Cref{fig:dedup}). 

Thus, data \emph{repetition} (e.g., due to poor preprocessing) can exacerbate user inference. However, the results on Reddit and Enron Emails (no duplicates) suggest that deduplication alone is insufficient to fully mitigate user inference.

\begin{figure}[h]
\captionsetup{width=0.45\linewidth} 
    \begin{minipage}[t]{0.5\linewidth}
        \centering
        \includegraphics[width=\linewidth]{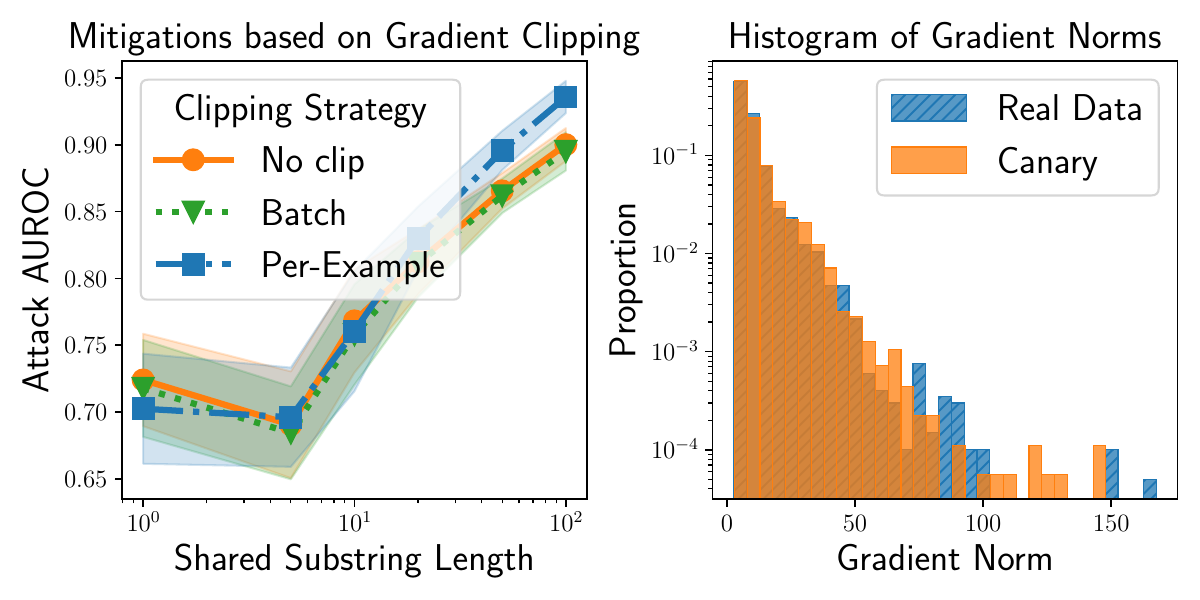}
        \caption{\small
Mitigation with gradient clipping. \textbf{Left}: Attack effectiveness for canaries with different shared substring lengths with gradient clipping ($125$M model, CC News). \textbf{Right}: The distribution of gradient norms for canary examples and real examples.}
        \label{fig:gradient_clipping}
    \end{minipage}
    \begin{minipage}[t]{0.5\linewidth}
        \centering
        \includegraphics[width=0.85\linewidth]{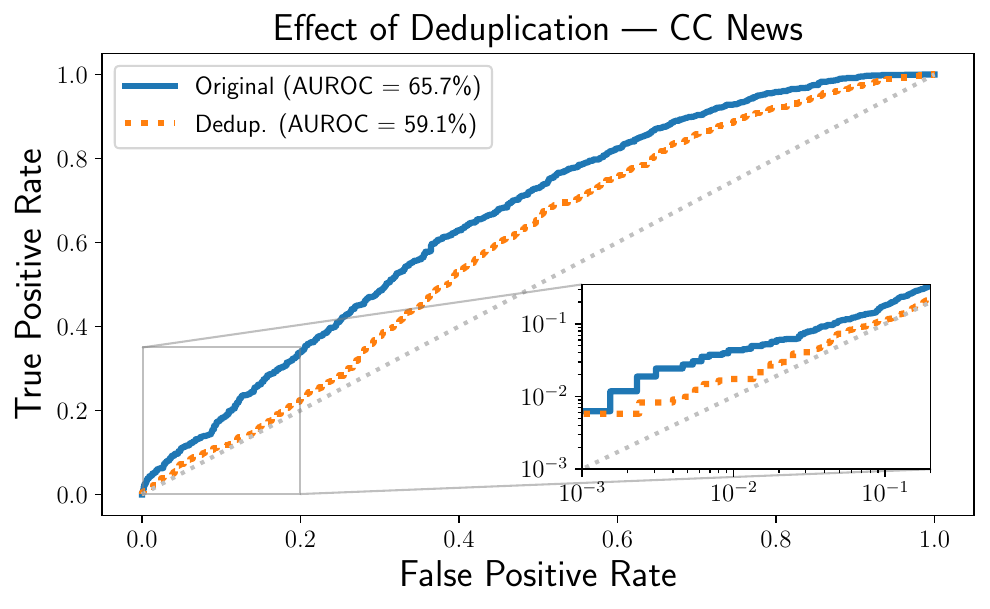}
        \caption{\small Effect of data deduplication per-user on CC News. \Cref{tab:ccnews-dedup} in \Cref{appendix:expt-results} gives TPR values at low FPR.}
        \label{fig:dedup}
    \end{minipage}
\end{figure}

\myparagraph{Example-level Differential Privacy (DP)}
DP~\cite{DMNS} gives provable bounds on privacy leakage. We study how example-level DP, which protects the privacy of individual \emph{examples}, impacts \emph{user} inference. We train the 125M model on Enron Emails using DP-Adam, a variant of Adam that clips per-example gradients and adds noise calibrated to the privacy budget $\varepsilon$. We find next that example-level DP can somewhat mitigate user inference while incurring increased compute cost and a degraded model utility.

Obtaining good utility with DP requires large batches and more epochs~\citep{ponomareva2023dp}, so we use a batch size of $1024$, tune the learning rate, and train the model for $50$ epochs ($1.2K$ updates), so that each job runs in $24$h (in comparison, non-private training takes $1.5$h for $7$ epochs). Further details of the tuning are given in \Cref{appendix:dp-tuning}.

\Cref{tab:enron_dp} shows a severe degradation in the validation loss under DP. For instance, a loss of $2.67$ at the weak guarantee of  $\varepsilon=32$ is  surpassed after just $1/3$\textsuperscript{rd} of an epoch of non-private training; this loss continues to reduce to $2.43$ after $3$ epochs.
In terms of attack effectiveness, example-level DP reduces the attack AUROC and the TPR at FPR $=5\%$, while the TPR at FPR $=1\%$ remains the same or gets worse.
Indeed, while example-level DP protects individual examples, it can fail to protect the privacy of \emph{users}, especially when they contribute many examples.
This highlights the need for scalable algorithms and software for fine-tuning LLMs with \emph{DP at the user-level}. Currently,  user-level DP algorithms have been designed for small models in federated learning, but do not yet scale to LLMs.

\begin{table}[ht]
\renewcommand{\arraystretch}{1.2}
\centering
\adjustbox{max width=0.99\linewidth}{%
\small
\begin{tabular}{ccccc}
    \toprule
    \textbf{Metric} & $\varepsilon=2$ & $\varepsilon=8$ & $\varepsilon=32$ & Non-private \\
    \midrule
    Val. Loss & $2.77$ & $2.71$ & $2.67$ & $2.43$ \\
    Attack AUROC & $64.7\%$ & $66.7\%$ & $67.9\%$ & $88.1\%$ 
    \\
    TPR @ FPR$=1\%$ & $8.8\%$ &	$8.8\%$ & $10.3\%$ &  	 $4.4\%$ \\
    TPR @ FPR$=5\%$ & $11.8\%$ & $10.3\%$ &	$10.3\%$  &	$27.9\%$ \\
    \bottomrule
    \end{tabular}}
    \caption{\small \textbf{Example-level differential privacy}: Training a model on Enron Emails under $(\varepsilon, 10^{-6})$-DP at the example-level (smaller $\varepsilon$ implies a higher level of privacy).}
    \label{tab:enron_dp}
\end{table}

\myparagraph{Summary}
Our results show that user inference is hard to mitigate with common heuristics.
Careful deduplication is necessary to ensure that data repetition does not exacerbate user inference.
Enforcing data limits per user can be effective but this only works for data-rich applications with a large number of users. Example-level DP can offer moderate mitigation but at the cost of increased data/compute and degraded model utility.
Developing an effective mitigation strategy that also works efficiently in data-scarce applications remains an open problem.

\section{Discussion and Conclusion}

When collecting data for fine-tuning an LLM, data from a company's users is often the natural choice since it closely resembles the types of inputs a deployed LLM will encounter.
However fine-tuning on user-stratified data also exposes new opportunities for privacy leakage. 
Until now, most work on privacy of LLMs have ignored any structure in the training data, but as the field shifts towards collecting data from new, potentially sensitive, sources, it is important to adapt our privacy threat models accordingly. Our work introduces a novel privacy attack exposing user participation in fine-tuning, and future work should explore other LLM privacy violations beyond membership inference and training data extraction. Furthermore, this work underscores the need for scaling user-aware training pipelines, such as user-level DP, to handle large datasets and models.


\section{Broader Impacts}

This work highlights a novel privacy vulnerability in LLMs fine-tuned on potentially sensitive user data.  Hypothetically, our methods could be leveraged by an attacker with API access to a fine-tuned LLM to infer which users contributed their data to the model's fine-tuning set. To mitigate the risk of data exposure, we performed experiments on public GPT-Neo models, using public datasets for fine-tuning, ensuring that our experiments do not disclose any sensitive user information. 

We envision that these methods will offer practical tools for conducting privacy audits of LLMs before releasing them for public use. By running user inference attacks, a company fine-tuning LLMs on user data can gain insights into the privacy risks exposed by providing access to the models and assess the effectiveness  of deploying mitigations. To counteract our proposed attacks, we evaluate several defense strategies, including example-level differential privacy and restricting individual user contributions, both of which provide partial mitigation of this threat.  We leave to future work the challenging problem of fully protecting LLMs against user inference with provable guarantees. 

\bibliography{refs}

\begin{thebibliography}{67}
\providecommand{\natexlab}[1]{#1}
\providecommand{\url}[1]{\texttt{#1}}
\expandafter\ifx\csname urlstyle\endcsname\relax
  \providecommand{\doi}[1]{doi: #1}\else
  \providecommand{\doi}{doi: \begingroup \urlstyle{rm}\Url}\fi

\bibitem[Abadi et~al.(2016)Abadi, Chu, Goodfellow, McMahan, Mironov, Talwar,
  and Zhang]{Abadi_2016}
M.~Abadi, A.~Chu, I.~Goodfellow, H.~B. McMahan, I.~Mironov, K.~Talwar, and
  L.~Zhang.
\newblock Deep learning with differential privacy.
\newblock In \emph{Proceedings of the {ACM} {SIGSAC} Conference on Computer and
  Communications Security}, 2016.

\bibitem[Anil et~al.(2023)Anil, Dai, Firat, Johnson, Lepikhin, Passos, Shakeri,
  Taropa, Bailey, Chen, et~al.]{anil2023palm}
R.~Anil, A.~M. Dai, O.~Firat, M.~Johnson, D.~Lepikhin, A.~Passos, S.~Shakeri,
  E.~Taropa, P.~Bailey, Z.~Chen, et~al.
\newblock Palm 2 technical report.
\newblock \emph{arXiv:2305.10403}, 2023.

\bibitem[Baumgartner et~al.(2020)Baumgartner, Zannettou, Keegan, Squire, and
  Blackburn]{pushshift_reddit}
J.~Baumgartner, S.~Zannettou, B.~Keegan, M.~Squire, and J.~Blackburn.
\newblock The pushshift reddit dataset.
\newblock \emph{Proceedings of the International AAAI Conference on Web and
  Social Media}, 14\penalty0 (1):\penalty0 830--839, May 2020.
\newblock \doi{10.1609/icwsm.v14i1.7347}.
\newblock URL \url{https://ojs.aaai.org/index.php/ICWSM/article/view/7347}.

\bibitem[Biderman et~al.(2023)Biderman, Prashanth, Sutawika, Schoelkopf,
  Anthony, Purohit, and Raf]{biderman2023emergent}
S.~Biderman, U.~S. Prashanth, L.~Sutawika, H.~Schoelkopf, Q.~Anthony,
  S.~Purohit, and E.~Raf.
\newblock {Emergent and Predictable Memorization in Large Language Models}.
\newblock \emph{arXiv:2304.11158}, 2023.

\bibitem[Black et~al.(2021)Black, Gao, Wang, Leahy, and Biderman]{gpt-neo}
S.~Black, L.~Gao, P.~Wang, C.~Leahy, and S.~Biderman.
\newblock {GPT-Neo: Large Scale Autoregressive Language Modeling with
  Mesh-Tensorflow}, Mar. 2021.

\bibitem[Carlini et~al.(2021)Carlini, Tram{\`e}r, Wallace, Jagielski,
  Herbert-Voss, Lee, Roberts, Brown, Song, Erlingsson, Oprea, and
  Raffel]{Carlini_Memorization}
N.~Carlini, F.~Tram{\`e}r, E.~Wallace, M.~Jagielski, A.~Herbert-Voss, K.~Lee,
  A.~Roberts, T.~Brown, D.~Song, {\'U}.~Erlingsson, A.~Oprea, and C.~Raffel.
\newblock Extracting training data from large language models.
\newblock In \emph{USENIX}, 2021.

\bibitem[Carlini et~al.(2022)Carlini, Chien, Nasr, Song, Terzis, and
  Tramèr]{LiRA}
N.~Carlini, S.~Chien, M.~Nasr, S.~Song, A.~Terzis, and F.~Tramèr.
\newblock Membership inference attacks from first principles.
\newblock In \emph{IEEE Symposium on Security and Privacy}, 2022.

\bibitem[Carlini et~al.(2023)Carlini, Ippolito, Jagielski, Lee, Tramer, and
  Zhang]{carlini2023quantifying}
N.~Carlini, D.~Ippolito, M.~Jagielski, K.~Lee, F.~Tramer, and C.~Zhang.
\newblock Quantifying memorization across neural language models.
\newblock In \emph{ICLR}, 2023.

\bibitem[Caruana et~al.(2000)Caruana, Lawrence, and
  Giles]{caruana2000overfitting}
R.~Caruana, S.~Lawrence, and C.~Giles.
\newblock {Overfitting in Neural Nets: Backpropagation, Conjugate Gradient, and
  Early Stopping }.
\newblock \emph{NeurIPS}, 2000.

\bibitem[Charles et~al.(2023)Charles, Mitchell, Pillutla, Reneer, and
  Garrett]{charles2023towards}
Z.~Charles, N.~Mitchell, K.~Pillutla, M.~Reneer, and Z.~Garrett.
\newblock {Towards Federated Foundation Models: Scalable Dataset Pipelines for
  Group-Structured Learning}.
\newblock \emph{arXiv:2307.09619}, 2023.

\bibitem[Chen et~al.(2021)Chen, Tworek, Jun, Yuan, de~Oliveira~Pinto, Kaplan,
  Edwards, Burda, Joseph, Brockman, Ray, Puri, Krueger, Petrov, Khlaaf, Sastry,
  Mishkin, Chan, Gray, Ryder, Pavlov, Power, Kaiser, Bavarian, Winter, Tillet,
  Such, Cummings, Plappert, Chantzis, Barnes, Herbert-Voss, Guss, Nichol,
  Paino, Tezak, Tang, Babuschkin, Balaji, Jain, Saunders, Hesse, Carr, Leike,
  Achiam, Misra, Morikawa, Radford, Knight, Brundage, Murati, Mayer, Welinder,
  McGrew, Amodei, McCandlish, Sutskever, and Zaremba]{chen2021evaluating}
M.~Chen, J.~Tworek, H.~Jun, Q.~Yuan, H.~P. de~Oliveira~Pinto, J.~Kaplan,
  H.~Edwards, Y.~Burda, N.~Joseph, G.~Brockman, A.~Ray, R.~Puri, G.~Krueger,
  M.~Petrov, H.~Khlaaf, G.~Sastry, P.~Mishkin, B.~Chan, S.~Gray, N.~Ryder,
  M.~Pavlov, A.~Power, L.~Kaiser, M.~Bavarian, C.~Winter, P.~Tillet, F.~P.
  Such, D.~Cummings, M.~Plappert, F.~Chantzis, E.~Barnes, A.~Herbert-Voss,
  W.~H. Guss, A.~Nichol, A.~Paino, N.~Tezak, J.~Tang, I.~Babuschkin, S.~Balaji,
  S.~Jain, W.~Saunders, C.~Hesse, A.~N. Carr, J.~Leike, J.~Achiam, V.~Misra,
  E.~Morikawa, A.~Radford, M.~Knight, M.~Brundage, M.~Murati, K.~Mayer,
  P.~Welinder, B.~McGrew, D.~Amodei, S.~McCandlish, I.~Sutskever, and
  W.~Zaremba.
\newblock Evaluating large language models trained on code.
\newblock \emph{arXiv 2107.03374}, 2021.

\bibitem[Chen et~al.(2023)Chen, Zhang, Wang, Backes, and Zhang]{face_auditor}
M.~Chen, Z.~Zhang, T.~Wang, M.~Backes, and Y.~Zhang.
\newblock {FACE-AUDITOR}: Data auditing in facial recognition systems.
\newblock In \emph{32nd USENIX Security Symposium (USENIX Security 23)}, pages
  7195--7212, Anaheim, CA, Aug. 2023. USENIX Association.
\newblock ISBN 978-1-939133-37-3.
\newblock URL
  \url{https://www.usenix.org/conference/usenixsecurity23/presentation/chen-min}.

\bibitem[Chen et~al.(2019)Chen, Lee, Bansal, Cao, Zhang, Lu, Tsay, Wang, Dai,
  Chen, Sohn, and Wu]{dai2019gmail}
M.~X. Chen, B.~N. Lee, G.~Bansal, Y.~Cao, S.~Zhang, J.~Lu, J.~Tsay, Y.~Wang,
  A.~M. Dai, Z.~Chen, T.~Sohn, and Y.~Wu.
\newblock Gmail smart compose: Real-time assisted writing.
\newblock In \emph{Proceedings of the 25th ACM SIGKDD International Conference
  on Knowledge Discovery \& Data Mining}, 2019.

\bibitem[Choquette-Choo et~al.(2021)Choquette-Choo, Tramer, Carlini, and
  Papernot]{choquette2021label}
C.~A. Choquette-Choo, F.~Tramer, N.~Carlini, and N.~Papernot.
\newblock Label-only membership inference attacks.
\newblock In \emph{ICML}, 2021.

\bibitem[Clement et~al.(2019)Clement, Bierbaum, O'Keeffe, and
  Alemi]{clement2019arxiv}
C.~B. Clement, M.~Bierbaum, K.~P. O'Keeffe, and A.~A. Alemi.
\newblock On the use of arxiv as a dataset.
\newblock \emph{arXiv 1905.00075}, 2019.

\bibitem[Debenedetti et~al.(2023)Debenedetti, Severi, Carlini, Choquette-Choo,
  Jagielski, Nasr, Wallace, and Tram{\`e}r]{debenedetti2023privacy}
E.~Debenedetti, G.~Severi, N.~Carlini, C.~A. Choquette-Choo, M.~Jagielski,
  M.~Nasr, E.~Wallace, and F.~Tram{\`e}r.
\newblock Privacy side channels in machine learning systems.
\newblock \emph{arXiv:2309.05610}, 2023.

\bibitem[Dwork et~al.(2006)Dwork, McSherry, Nissim, and Smith]{DMNS}
C.~Dwork, F.~McSherry, K.~Nissim, and A.~Smith.
\newblock {Calibrating Noise to Sensitivity in Private Data Analysis}.
\newblock In \emph{Proc. of the Third Conf. on Theory of Cryptography (TCC)},
  pages 265--284, 2006.
\newblock URL \url{http://dx.doi.org/10.1007/11681878\_14}.

\bibitem[Ganju et~al.(2018)Ganju, Wang, Yang, Gunter, and
  Borisov]{Ganju_PropInf_2018}
K.~Ganju, Q.~Wang, W.~Yang, C.~A. Gunter, and N.~Borisov.
\newblock {Property Inference Attacks on Fully Connected Neural Networks Using
  Permutation Invariant Representations}.
\newblock In \emph{Proceedings of the 2018 ACM SIGSAC Conference on Computer
  and Communications Security}, page 619–633, 2018.

\bibitem[Gao et~al.(2020)Gao, Biderman, Black, Golding, Hoppe, Foster, Phang,
  He, Thite, Nabeshima, et~al.]{gao2020pile}
L.~Gao, S.~Biderman, S.~Black, L.~Golding, T.~Hoppe, C.~Foster, J.~Phang,
  H.~He, A.~Thite, N.~Nabeshima, et~al.
\newblock The pile: An 800gb dataset of diverse text for language modeling.
\newblock \emph{arXiv:2101.00027}, 2020.

\bibitem[Haim et~al.(2022)Haim, Vardi, Yehudai, michal Irani, and
  Shamir]{haim2022reconstructing}
N.~Haim, G.~Vardi, G.~Yehudai, michal Irani, and O.~Shamir.
\newblock Reconstructing training data from trained neural networks.
\newblock In \emph{NeurIPS}, 2022.

\bibitem[Hamborg et~al.(2017)Hamborg, Meuschke, Breitinger, and
  Gipp]{Hamborg2017}
F.~Hamborg, N.~Meuschke, C.~Breitinger, and B.~Gipp.
\newblock news-please: A generic news crawler and extractor.
\newblock In \emph{Proceedings of the 15th International Symposium of
  Information Science}, 2017.

\bibitem[Hartmann et~al.(2023)Hartmann, Meynent, Peyrard, Dimitriadis, Tople,
  and West]{Hartmann_DistrInf_2023}
V.~Hartmann, L.~Meynent, M.~Peyrard, D.~Dimitriadis, S.~Tople, and R.~West.
\newblock {Distribution Inference Risks: Identifying and Mitigating Sources of
  Leakage}.
\newblock In \emph{2023 IEEE Conference on Secure and Trustworthy Machine
  Learning (SaTML)}, pages 136--149, 2023.

\bibitem[Inan et~al.(2021)Inan, Ramadan, Wutschitz, Jones, Rühle, Withers, and
  Sim]{inan2021training}
H.~A. Inan, O.~Ramadan, L.~Wutschitz, D.~Jones, V.~Rühle, J.~Withers, and
  R.~Sim.
\newblock Training data leakage analysis in language models.
\newblock \emph{arxiv:2101.05405}, 2021.

\bibitem[Ippolito et~al.(2023)Ippolito, Tramer, Nasr, Zhang, Jagielski, Lee,
  Choquette~Choo, and Carlini]{ippolito2022preventing}
D.~Ippolito, F.~Tramer, M.~Nasr, C.~Zhang, M.~Jagielski, K.~Lee,
  C.~Choquette~Choo, and N.~Carlini.
\newblock Preventing generation of verbatim memorization in language models
  gives a false sense of privacy.
\newblock In \emph{INLG}, 2023.

\bibitem[Jagielski et~al.(2023{\natexlab{a}})Jagielski, Nasr, Choquette-Choo,
  Lee, and Carlini]{jagielski2023students}
M.~Jagielski, M.~Nasr, C.~Choquette-Choo, K.~Lee, and N.~Carlini.
\newblock Students parrot their teachers: Membership inference on model
  distillation.
\newblock \emph{arXiv:2303.03446}, 2023{\natexlab{a}}.

\bibitem[Jagielski et~al.(2023{\natexlab{b}})Jagielski, Thakkar, Tramer,
  Ippolito, Lee, Carlini, Wallace, Song, Thakurta, Papernot, and
  Zhang]{jagielski2022measuring}
M.~Jagielski, O.~Thakkar, F.~Tramer, D.~Ippolito, K.~Lee, N.~Carlini,
  E.~Wallace, S.~Song, A.~G. Thakurta, N.~Papernot, and C.~Zhang.
\newblock Measuring forgetting of memorized training examples.
\newblock In \emph{ICLR}, 2023{\natexlab{b}}.

\bibitem[Kairouz et~al.(2015)Kairouz, Oh, and
  Viswanath]{kairouz2015composition}
P.~Kairouz, S.~Oh, and P.~Viswanath.
\newblock {The Composition Theorem for Differential Privacy}.
\newblock In \emph{ICML}, pages 1376--1385, 2015.

\bibitem[Kairouz et~al.(2021)Kairouz, McMahan, Song, Thakkar, Thakurta, and
  Xu]{kairouz2021practical}
P.~Kairouz, B.~McMahan, S.~Song, O.~Thakkar, A.~Thakurta, and Z.~Xu.
\newblock Practical and private (deep) learning without sampling or shuffling.
\newblock In \emph{ICML}, 2021.

\bibitem[Kandpal et~al.(2022)Kandpal, Wallace, and
  Raffel]{kandpal2022deduplicating}
N.~Kandpal, E.~Wallace, and C.~Raffel.
\newblock Deduplicating training data mitigates privacy risks in language
  models.
\newblock In \emph{ICML}, 2022.

\bibitem[Kingma and Ba(2015)]{kingma2017adam}
D.~P. Kingma and J.~Ba.
\newblock {Adam: {A} Method for Stochastic Optimization}.
\newblock In \emph{ICLR}, 2015.

\bibitem[Klimt and Yang(2004)]{Klimt2004IntroducingTE}
B.~Klimt and Y.~Yang.
\newblock Introducing the enron corpus.
\newblock In \emph{International Conference on Email and Anti-Spam}, 2004.

\bibitem[Kudugunta et~al.(2023)Kudugunta, Caswell, Zhang, Garcia,
  Choquette-Choo, Lee, Xin, Kusupati, Stella, Bapna,
  et~al.]{kudugunta2023madlad}
S.~Kudugunta, I.~Caswell, B.~Zhang, X.~Garcia, C.~A. Choquette-Choo, K.~Lee,
  D.~Xin, A.~Kusupati, R.~Stella, A.~Bapna, et~al.
\newblock Madlad-400: A multilingual and document-level large audited dataset.
\newblock \emph{arXiv:2309.04662}, 2023.

\bibitem[Lee et~al.(2022)Lee, Ippolito, Nystrom, Zhang, Eck, Callison-Burch,
  and Carlini]{lee2021deduplicating}
K.~Lee, D.~Ippolito, A.~Nystrom, C.~Zhang, D.~Eck, C.~Callison-Burch, and
  N.~Carlini.
\newblock Deduplicating training data makes language models better.
\newblock In \emph{ACL}, 2022.

\bibitem[Lehmann et~al.(1986)Lehmann, Romano, and Casella]{lehmann1986testing}
E.~L. Lehmann, J.~P. Romano, and G.~Casella.
\newblock \emph{{Testing Statistical Hypotheses}}, volume~3.
\newblock Springer, 1986.

\bibitem[Levy et~al.(2021)Levy, Sun, Amin, Kale, Kulesza, Mohri, and
  Suresh]{levy2021learning}
D.~A.~N. Levy, Z.~Sun, K.~Amin, S.~Kale, A.~Kulesza, M.~Mohri, and A.~T.
  Suresh.
\newblock Learning with user-level privacy.
\newblock In \emph{NeurIPS}, 2021.

\bibitem[Li et~al.(2022)Li, Rezaei, and Liu]{li2022userlevel}
G.~Li, S.~Rezaei, and X.~Liu.
\newblock {User-Level Membership Inference Attack against Metric Embedding
  Learning}.
\newblock In \emph{ICLR 2022 Workshop on PAIR2Struct: Privacy, Accountability,
  Interpretability, Robustness, Reasoning on Structured Data}, 2022.

\bibitem[Liu et~al.(2022)Liu, Tam, Mohammed, Mohta, Huang, Bansal, and
  Raffel]{liu2022fewshot}
H.~Liu, D.~Tam, M.~Mohammed, J.~Mohta, T.~Huang, M.~Bansal, and C.~Raffel.
\newblock Few-shot parameter-efficient fine-tuning is better and cheaper than
  in-context learning.
\newblock In \emph{NeurIPS}, 2022.

\bibitem[Lukas et~al.(2023)Lukas, Salem, Sim, Tople, Wutschitz, and
  Zanella-Beguelin]{pii_leakage}
N.~Lukas, A.~Salem, R.~Sim, S.~Tople, L.~Wutschitz, and S.~Zanella-Beguelin.
\newblock Analyzing leakage of personally identifiable information in language
  models.
\newblock In \emph{IEEE Symposium on Security and Privacy}, 2023.

\bibitem[Luyckx and Daelemans(2008)]{luyckx-daelemans-2008-authorship}
K.~Luyckx and W.~Daelemans.
\newblock Authorship attribution and verification with many authors and limited
  data.
\newblock In D.~Scott and H.~Uszkoreit, editors, \emph{Proceedings of the 22nd
  International Conference on Computational Linguistics (Coling 2008)}, pages
  513--520, Manchester, UK, Aug. 2008. Coling 2008 Organizing Committee.
\newblock URL \url{https://aclanthology.org/C08-1065}.

\bibitem[Luyckx and Daelemans(2010)]{10.1093/llc/fqq013}
K.~Luyckx and W.~Daelemans.
\newblock {The effect of author set size and data size in authorship
  attribution}.
\newblock \emph{Literary and Linguistic Computing}, 26\penalty0 (1):\penalty0
  35--55, 08 2010.
\newblock ISSN 0268-1145.
\newblock \doi{10.1093/llc/fqq013}.
\newblock URL \url{https://doi.org/10.1093/llc/fqq013}.

\bibitem[Mattern et~al.(2023)Mattern, Mireshghallah, Jin, Schoelkopf, Sachan,
  and Berg-Kirkpatrick]{mattern-etal-2023-membership}
J.~Mattern, F.~Mireshghallah, Z.~Jin, B.~Schoelkopf, M.~Sachan, and
  T.~Berg-Kirkpatrick.
\newblock Membership inference attacks against language models via
  neighbourhood comparison.
\newblock In \emph{Findings of ACL}, 2023.

\bibitem[McMahan et~al.(2018)McMahan, Ramage, Talwar, and
  Zhang]{mcmahan2017learning}
H.~B. McMahan, D.~Ramage, K.~Talwar, and L.~Zhang.
\newblock Learning differentially private recurrent language models.
\newblock In \emph{International Conference on Learning Representations}, 2018.

\bibitem[Miao et~al.(2021)Miao, Xue, Chen, Pan, Zhang, Zhao, Kaafar, and
  Xiang]{miao2021audio}
Y.~Miao, M.~Xue, C.~Chen, L.~Pan, J.~Zhang, B.~Z.~H. Zhao, D.~Kaafar, and
  Y.~Xiang.
\newblock {The Audio Auditor: User-Level Membership Inference in Internet of
  Things Voice Services}.
\newblock In \emph{Privacy Enhancing Technologies Symposium (PETS)}, 2021.

\bibitem[Mireshghallah et~al.(2022)Mireshghallah, Goyal, Uniyal,
  Berg-Kirkpatrick, and Shokri]{mireshghallah-etal-2022-quantifying}
F.~Mireshghallah, K.~Goyal, A.~Uniyal, T.~Berg-Kirkpatrick, and R.~Shokri.
\newblock Quantifying privacy risks of masked language models using membership
  inference attacks.
\newblock In \emph{EMNLP}, 2022.

\bibitem[Mosbach et~al.(2023)Mosbach, Pimentel, Ravfogel, Klakow, and
  Elazar]{Mosbach2023FewshotFV}
M.~Mosbach, T.~Pimentel, S.~Ravfogel, D.~Klakow, and Y.~Elazar.
\newblock Few-shot fine-tuning vs. in-context learning: A fair comparison and
  evaluation.
\newblock In \emph{Findings of ACL}, 2023.

\bibitem[Nasr et~al.(2023)Nasr, Carlini, Hayase, Jagielski, Cooper, Ippolito,
  Choquette-Choo, Wallace, Tram{\`e}r, and Lee]{nasr2023scalable}
M.~Nasr, N.~Carlini, J.~Hayase, M.~Jagielski, A.~F. Cooper, D.~Ippolito, C.~A.
  Choquette-Choo, E.~Wallace, F.~Tram{\`e}r, and K.~Lee.
\newblock Scalable extraction of training data from (production) language
  models.
\newblock \emph{arXiv preprint arXiv:2311.17035}, 2023.

\bibitem[Niu et~al.(2023)Niu, Mirza, Maradni, and P{\"o}pper]{code_leakage}
L.~Niu, S.~Mirza, Z.~Maradni, and C.~P{\"o}pper.
\newblock {CodexLeaks}: Privacy leaks from code generation language models in
  {GitHub} copilot.
\newblock In \emph{USENIX Security Symposium}, 2023.

\bibitem[Oprea and Vassilev(2023)]{NIST_report}
A.~Oprea and A.~Vassilev.
\newblock Adversarial machine learning: A taxonomy and terminology of attacks
  and mitigations.
\newblock NIST AI 100-2 E2023 report. Available at
  \url{https://csrc.nist.gov/pubs/ai/100/2/e2023/ipd}, 2023.

\bibitem[Pascanu et~al.(2013)Pascanu, Mikolov, and
  Bengio]{pascanu2013difficulty}
R.~Pascanu, T.~Mikolov, and Y.~Bengio.
\newblock On the difficulty of training recurrent neural networks.
\newblock In \emph{ICML}, 2013.

\bibitem[Ponomareva et~al.(2023)Ponomareva, Hazimeh, Kurakin, Xu, Denison,
  McMahan, Vassilvitskii, Chien, and Thakurta]{ponomareva2023dp}
N.~Ponomareva, H.~Hazimeh, A.~Kurakin, Z.~Xu, C.~Denison, H.~B. McMahan,
  S.~Vassilvitskii, S.~Chien, and A.~G. Thakurta.
\newblock {How to DP-fy ML: A Practical Guide to Machine Learning with
  Differential Privacy}.
\newblock \emph{Journal of Artificial Intelligence Research}, 77:\penalty0
  1113--1201, 2023.

\bibitem[Radford et~al.(2019)Radford, Wu, Child, Luan, Amodei, and
  Sutskever]{Radford2019LanguageMA}
A.~Radford, J.~Wu, R.~Child, D.~Luan, D.~Amodei, and I.~Sutskever.
\newblock Language models are unsupervised multitask learners.
\newblock 2019.

\bibitem[Ramaswamy et~al.(2020)Ramaswamy, Thakkar, Mathews, Andrew, McMahan,
  and Beaufays]{ramaswamy2020training}
S.~Ramaswamy, O.~Thakkar, R.~Mathews, G.~Andrew, H.~B. McMahan, and
  F.~Beaufays.
\newblock Training production language models without memorizing user data.
\newblock \emph{arxiv:2009.10031}, 2020.

\bibitem[Reddi et~al.(2021)Reddi, Charles, Zaheer, Garrett, Rush,
  Kone{\v{c}}n{\'y}, Kumar, and McMahan]{reddi2021adaptive}
S.~J. Reddi, Z.~Charles, M.~Zaheer, Z.~Garrett, K.~Rush, J.~Kone{\v{c}}n{\'y},
  S.~Kumar, and H.~B. McMahan.
\newblock {Adaptive Federated Optimization}.
\newblock In \emph{ICLR}, 2021.

\bibitem[Sablayrolles et~al.(2019)Sablayrolles, Douze, Schmid, Ollivier, and
  J{\'e}gou]{sablayrolles2019whitebox}
A.~Sablayrolles, M.~Douze, C.~Schmid, Y.~Ollivier, and H.~J{\'e}gou.
\newblock White-box vs black-box: Bayes optimal strategies for membership
  inference.
\newblock In \emph{ICML}, 2019.

\bibitem[Sainz et~al.(2023)Sainz, Campos, García-Ferrero, Etxaniz, and
  Agirre]{sainz}
O.~Sainz, J.~A. Campos, I.~García-Ferrero, J.~Etxaniz, and E.~Agirre.
\newblock {Did ChatGPT Cheat on Your Test?}
\newblock \url{https://hitz-zentroa.github.io/lm-contamination/blog/}, 2023.

\bibitem[Shejwalkar et~al.(2021)Shejwalkar, Inan, Houmansadr, and
  Sim]{shejwalkar2021membership}
V.~Shejwalkar, H.~A. Inan, A.~Houmansadr, and R.~Sim.
\newblock {Membership Inference Attacks Against NLP Classification Models}.
\newblock In \emph{NeurIPS 2021 Workshop Privacy in Machine Learning}, 2021.

\bibitem[Shokri et~al.(2017)Shokri, Stronati, Song, and
  Shmatikov]{shokri2017membership}
R.~Shokri, M.~Stronati, C.~Song, and V.~Shmatikov.
\newblock Membership inference attacks against machine learning models.
\newblock In \emph{IEEE Symposium on Security and Privacy}, 2017.

\bibitem[Song and Shmatikov(2019)]{song2019auditing}
C.~Song and V.~Shmatikov.
\newblock Auditing data provenance in text-generation models.
\newblock In \emph{Proceedings of the 25th ACM SIGKDD International Conference
  on Knowledge Discovery \& Data Mining}, 2019.

\bibitem[Song et~al.(2020)Song, Wang, Zhang, Song, Wang, Ren, and
  Qi]{Song_FL_User_Privacy_2020}
M.~Song, Z.~Wang, Z.~Zhang, Y.~Song, Q.~Wang, J.~Ren, and H.~Qi.
\newblock {Analyzing User-Level Privacy Attack Against Federated Learning}.
\newblock \emph{IEEE Journal on Selected Areas in Communications}, 38\penalty0
  (10):\penalty0 2430--2444, 2020.

\bibitem[Tirumala et~al.(2022)Tirumala, Markosyan, Zettlemoyer, and
  Aghajanyan]{NEURIPS2022_fa0509f4}
K.~Tirumala, A.~Markosyan, L.~Zettlemoyer, and A.~Aghajanyan.
\newblock Memorization without overfitting: Analyzing the training dynamics of
  large language models.
\newblock In \emph{NeurIPS}, 2022.

\bibitem[Wang et~al.(2019)Wang, Song, Zhang, Song, Wang, and
  Qi]{Wang_FL_User_2019}
Z.~Wang, M.~Song, Z.~Zhang, Y.~Song, Q.~Wang, and H.~Qi.
\newblock {Beyond Inferring Class Representatives: User-Level Privacy Leakage
  From Federated Learning}.
\newblock In \emph{IEEE INFOCOM 2019 - IEEE Conference on Computer
  Communications}, page 2512–2520, 2019.

\bibitem[Watson et~al.(2022)Watson, Guo, Cormode, and
  Sablayrolles]{watson2022importance}
L.~Watson, C.~Guo, G.~Cormode, and A.~Sablayrolles.
\newblock On the importance of difficulty calibration in membership inference
  attacks.
\newblock In \emph{ICLR}, 2022.

\bibitem[Xu et~al.(2023)Xu, Zhang, Andrew, Choquette, Kairouz, Mcmahan,
  Rosenstock, and Zhang]{xu2023federated}
Z.~Xu, Y.~Zhang, G.~Andrew, C.~Choquette, P.~Kairouz, B.~Mcmahan,
  J.~Rosenstock, and Y.~Zhang.
\newblock Federated learning of gboard language models with differential
  privacy.
\newblock In \emph{ACL}, 2023.

\bibitem[Ye et~al.(2022)Ye, Maddi, Murakonda, Bindschaedler, and
  Shokri]{enhanced_mi}
J.~Ye, A.~Maddi, S.~K. Murakonda, V.~Bindschaedler, and R.~Shokri.
\newblock Enhanced membership inference attacks against machine learning
  models.
\newblock In \emph{Proceedings of the ACM SIGSAC Conference on Computer and
  Communications Security}, 2022.

\bibitem[Yeom et~al.(2018)Yeom, Giacomelli, Fredrikson, and Jha]{yeom_mi}
S.~Yeom, I.~Giacomelli, M.~Fredrikson, and S.~Jha.
\newblock Privacy risk in machine learning: Analyzing the connection to
  overfitting.
\newblock In \emph{IEEE Computer Security Foundations Symposium}, 2018.

\bibitem[Zhang et~al.(2021)Zhang, Ippolito, Lee, Jagielski, Tramèr, and
  Carlini]{zhang2021counterfactual}
C.~Zhang, D.~Ippolito, K.~Lee, M.~Jagielski, F.~Tramèr, and N.~Carlini.
\newblock Counterfactual memorization in neural language models.
\newblock \emph{arXiv 2112.12938}, 2021.

\bibitem[Zhang et~al.(2022)Zhang, Roller, Goyal, Artetxe, Chen, Chen, Dewan,
  Diab, Li, Lin, Mihaylov, Ott, Shleifer, Shuster, Simig, Koura, Sridhar, Wang,
  and Zettlemoyer]{zhang2022opt}
S.~Zhang, S.~Roller, N.~Goyal, M.~Artetxe, M.~Chen, S.~Chen, C.~Dewan, M.~Diab,
  X.~Li, X.~V. Lin, T.~Mihaylov, M.~Ott, S.~Shleifer, K.~Shuster, D.~Simig,
  P.~S. Koura, A.~Sridhar, T.~Wang, and L.~Zettlemoyer.
\newblock Opt: Open pre-trained transformer language models.
\newblock \emph{arXiv 2205.01068}, 2022.

\end{thebibliography}
\bibliographystyle{abbrvnat}

\clearpage

\appendix
\part{Appendix} 
The outline of the appendix is as follows:
\begin{itemize}[nosep]
    \item \Cref{appendix:theory}: Proof of the analysis of the attack statistic (\Cref{thm:analysis}).
    \item \Cref{appendix:other_methods}: Alternate approaches to solving user inference (e.g. if the computational cost was not a limiting factor).
    \item \Cref{sec:a:related}: Further details on related work.
    \item \Cref{appendix:experiments}: Detailed experimental setup (datasets, models, hyperparameters).
    \item \Cref{appendix:expt-results}: 
    Additional experimental results.
    \item \Cref{appendix:user-level-dp}:
    A discussion of user-level DP, its  promises, and challenges.
\end{itemize}

\section{Theoretical Analysis of the Attack Statistic} \label{appendix:theory}

We prove \Cref{thm:analysis} here.

\myparagraph{Recall of definitions}
The KL and $\chi^2$ divergences are defined respectively as
\[
    \kl(P \Vert Q) = \sum_{\xv} P(\xv) \log\left(\frac{P(\xv)}{Q(\xv)}\right)\,
    \quad\text{and}\quad
    \chi^2(P \Vert Q) = \sum_{\xv} \frac{P(\xv)^2}{Q(\xv)} - 1 \,.
\]

Recall that we also defined
\begin{gather*}
    \pRef(\xv) = \Puser_{-u}(\xv) = \frac{\sum_{u' \neq u} \alpha_{u'} \Puser_{u'}}{\sum_{u' \neq u} \alpha_{u'}} 
    = \frac{\sum_{u' \neq u} \alpha_{u'} \Puser_{u'}}{1 - \alpha_{u}}\,,
    \quad \text{and} \\
    p_\theta(\xv) = \sum_{u' = 1}^n \alpha_{u'} \Puser_{u'}(\xv) = \alpha_u \Puser_u(\xv) + (1 - \alpha_u) \Puser_{-u}(\xv) \,.
\end{gather*}

\myparagraph{Proof of the upper bound}
Using the inequality $\log(1+t) \le t$ we get, 
\begin{align*}
    \bar T(\Puser_u) 
    &= \expect_{\xv \sim \Puser_u} \left[ \log\left( \frac{p_\theta(\xv)}{\pRef(\xv)} \right) \right] \\
    &= \expect_{\xv \sim \Puser_u} \left[ \log\left( \frac{\alpha_u \Puser_u(\xv) + (1 - \alpha_u)\Puser_{-u}(\xv)}{\Puser_{-u}(\xv)} \right) \right] \\
    &= \expect_{\xv \sim \Puser_u}\left[
        \log\left(1 + \alpha_u \left( \tfrac{\Puser_u(\xv)}{\Puser_{-u}(\xv)} - 1 \right) \right)
    \right] \\
    &\le \alpha_{u} \, \expect_{\xv \sim \Puser_u}\left[\frac{\Puser_u(\xv)}{\Puser_{-u}(\xv)} - 1 \right] = \alpha_u \, \chi^2 \left(\Puser_u \Vert \Puser_{-u} \right) \,.
\end{align*}

\myparagraph{Proof of the lower bound}
Using $\log(1+t) > \log(t)$, we get
\begin{align*}
    \bar T(\Puser_u)
    &= \expect_{\xv \sim \Puser_u} \left[ \log\left( 
    \frac{p_\theta(\xv)}{\pRef(\xv)} \right) \right] 
    \\
    &= \expect_{\xv \sim \Puser_u} \left[ \log\left( \frac{\alpha_u \Puser_u(\xv) + (1 - \alpha_u)\Puser_{-u}(\xv)}{\Puser_{-u}(\xv)} \right) \right]
    \\
    &= \log(1 - \alpha_u) + \expect_{\xv \sim \Puser_u} \left[ \log\left( \frac{\alpha_u \Puser_u(\xv)}{(1 - \alpha_u)\Puser_{-u}(\xv)} + 1  \right) \right]
    \\
    &> \log(1 - \alpha_u) + \expect_{\xv \sim \Puser_u} \left[ \log\left( \frac{\alpha_u \Puser_u(\xv)}{(1 - \alpha_u)\Puser_{-u}(\xv)}\right) \right]
    \\
    &= \log(\alpha_u) + \expect_{\xv \sim \Puser_{u}} \left[ \log \left( \frac{\Puser_{u}(\xv)}{\Puser_{-u}(\xv)} \right) \right] = \log(\alpha_u) + \kl(\Puser_{u} \Vert \Puser_{-u}) \,.
\end{align*}

\section{Alternate Approaches to User Inference}
\label{appendix:other_methods}

We consider some alternate approaches to user inference that are inspired by the existing literature on membership inference. 
As we shall see, these approaches are impractical for the LLM user inference setting where exact samples from the fine-tuning data are not known to the attacker and models are costly to train.

A common approach for membership inference is to train ``shadow models'', models trained in a similar fashion and on similar data to the model being attacked \citep{shokri2017membership}. Once many shadow models have been trained, one can construct a classifier that identifies whether the target model has been trained on a particular example. Typically, this classifier takes as input a model's loss on the example in question and is learned based on the shadow models' losses on examples that were (or were not) a part of \emph{their} training data. This approach could in principle be adapted to user inference on LLMs.

First, we would need to assume that the attacker has enough data from user $u$ to fine-tune shadow models on datasets containing user $u$'s data as well as an additional set of samples used to compute $u$'s likelihood under the shadow models. Thus, we assume the attacker has $n$ samples $\xv_{train}\pow{1:n} := (\xv\pow{1}, \ldots, \xv\pow{n}) \sim \Puser_u^n$ used for shadow model training and $m$ samples $\xv\pow{1:m} := (\xv\pow{1}, \ldots, \xv\pow{m}) \sim \Puser_u^m$ used to compute likelihoods.

Next, the attacker trains many shadow models on data similar to the target model's fine-tuning data, including $\xv_{train}\pow{1:n}$ in half of the shadow models' fine-tuning data. This repeated training yields samples from two distributions: the distribution of models trained with user $u$'s data $\mathcal{P}$ and the distribution of models trained without user $u$'s data $\mathcal{Q}$. The goal of the user inference attack is to determine which distribution the target model is more likely sampled from.

However, since we assume the attacker has only black-box access to the target model, they must instead perform a different hypothesis test based on the likelihood of $\xv\pow{1:m}$ under the target model. To this end, the attacker must evaluate  the shadow models on $\xv\pow{1:m}$ to draw samples from:


\begin{align}
\begin{aligned}
    \mathcal{P}'\, &: \, p_\theta(\xv) \, \,  \text{where} \, \, \theta \sim \mathcal{P}, \xv \sim \Puser_u \,, \qquad
    \mathcal{Q}'\, :\, p_\theta(\xv) \, \,  \text{where} \, \, \theta \sim \mathcal{Q}, \xv \sim \Puser_u \,.
\end{aligned}
\end{align}

Finally, the attacker can classify user $u$ as being part (or not part) of the target model's fine-tuning data based on whether the likelihood values of the target model on $\xv\pow{1:m}$ are more likely under $\mathcal{P}'$ or $\mathcal{Q}'$.

While this is the ideal approach to performing user inference with no computational constraints, it is infeasible due to the cost of repeatedly training shadow LLMs and the assumption that the attacker has enough data from user $u$ to both train and evaluate shadow models.
\section{Further Details on Related Work}
\label{sec:a:related}

There are several papers that study the risk of user inference attacks, but they either have a different threat model, or are not applicable to LLMs.

\myparagraph{User-level Membership Inference}
We refer to problems of identifying a user's participation in training when given the exact training samples of the user as \emph{user-level membership inference}. 
\citet{song2019auditing} propose methods for inferring whether a user's data was part of the training set of a language model, under the assumption that the attacker has access to the user's training set. For their attack, they train multiple shadow models on subsets of multiple users' training data and a meta-classifier to distinguish users who participating in training from those who did not. This meta-classifier based  methodology  is not feasible for LLMs due to its high computational complexity.
Moreover, the notion of a ``user'' in their experiments is a random i.i.d. subset of the dataset; this does not work for the more realistic threat model of user inference, which relies on the similarity between the attacker's samples of a user to the training samples contributed by this user.

\citet{shejwalkar2021membership} also assume that the attacker knows the user's training set and perform user-level inference for NLP classification models by aggregating the results of membership inference for each sample of the target user.

\myparagraph{User Inference}
In the context of classification and regression, \citet{Hartmann_DistrInf_2023} define distributional membership inference, with the goal of identifying if a user participated in the training set of a model without knowledge of the exact training samples. 
This coincides with our definition of user inference.
\citet{Hartmann_DistrInf_2023} use existing shadow model-based attacks for distribution (or property) inference~\cite{Ganju_PropInf_2018}, as their main goal is to analyze sources of leakage and evaluate defenses. User inference attacks have been also studied in other applications domains, such as embedding learning for vision~\cite{li2022userlevel}  and speech recognition for IoT devices~\cite{miao2021audio}. \citet{face_auditor} design a black-box user-level auditing procedure on face recognition systems in which an auditor has access to images of a particular user that are not part of the training set. In federated learning, \citet{Wang_FL_User_2019} and \citet{Song_FL_User_Privacy_2020} analyze the risk of user inference by a malicious server.
None of these works apply to our LLM setting because they are either (a) domain-specific, or (b) computationally inefficient (e.g. due to shadow models).




\myparagraph{Comparison to Related Tasks}
User inference on text models is related to, but distinct from authorship attribution, the task of identifying authors from a user population given access to multiple writing samples.
We recall it definition and discuss the similarities and differences.
The goal of authorship attribution (AA) is to find which of the given population of users wrote a given text. For user inference (UI), on the other hand, the goal is to figure out \emph{if} any data from a given user was used to train a given \emph{model}. 
Note the key distinction here: there is no model in the problem statement of AA while the entire population of users is not assumed to be known for UI. Indeed, UI cannot be reduced to AA or vice versa: Solving AA does not solve UI because it does not tell us whether the user’s data was used to train a given LLM (which is absent from the problem statement of AA). Likewise, solving UI only tells us that a user’s data was used to train a given model but it does not tell us which user from a given population this data comes from (since the full population of users is not assumed to be known for UI).


Author attribution assumes that the entire user population is known, which is not required in user inference. Existing work on author attribution~\cite[e.g.][]{luyckx-daelemans-2008-authorship,10.1093/llc/fqq013} casts the problem as a classification task with one class per user, and does not scale to large number of users. 
Interestingly, \citet{10.1093/llc/fqq013} identified that the number of  authors and the amount of training data  per  author are important factors for the success of author attribution, also reflected by our findings when analyzing the user inference attack success. Connecting author attribution with privacy attacks on LLM fine-tuning could be a topic of future work.

\section{Experimental Setup} \label{appendix:experiments}

In this section, we give the following details:
\begin{itemize}[nosep]
    \item \Cref{sec:a:expt:ds}: Full details of the datasets, their preprocessing, the models used, and the evaluation of the attack.
    \item \Cref{sec:a:expt:canary}: Pseudocode of the canary construction algorithm.
    \item \Cref{sec:a:expt:mitigation}: Precise definitions of mitigation strategies.
    \item \Cref{appendix:dp-tuning}: Details of hyperparameter tuning for example-level DP.
    \item \Cref{sec:ccnews-dup}: 
    Analysis of the duplicates present in CC News.
\end{itemize}

\begin{figure*}[t]
    \includegraphics[width=\linewidth]{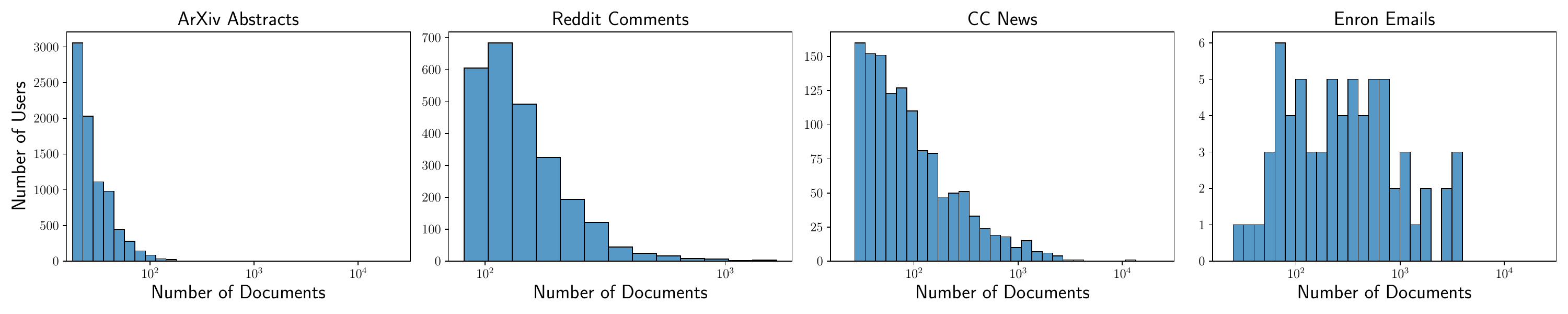}
\caption{\small Histogram of number of documents per user for each dataset.}
\label{fig:dataset_hist}
\end{figure*}

\subsection{Datasets, Models, Evaluation}
\label{sec:a:expt:ds}

We evaluate user inference attacks on four user-stratified datasets. Here, we describe the datasets, the notion of a ```user''' in each dataset, and any initial filtering steps applied. \Cref{fig:dataset_hist} gives a histogram of data per user (see also \Cref{tab:datasets,tab:datasets-arxiv}).

\begin{itemize}
 \item \textbf{Reddit Comments}\footnote{\url{https://huggingface.co/datasets/fddemarco/pushshift-reddit-comments}} \citep{pushshift_reddit} : Each example is a comment posted on Reddit. We define a user associated with a comment to be the username that posted the comment.
 
 The raw comment dump contains about 1.8 billion comments posted over a four-year span between 2012 and 2016. To make the dataset suitable for experiments on user inference, we take the following preprocessing steps:
 \begin{itemize}
     \item To reduce the size of the dataset, we initially filter to comments made during a six-month period between September 2015 and February 2016, resulting in a smaller dataset of 331 million comments.
     \item As a heuristic for filtering automated Reddit bot and moderator accounts from the dataset, we remove any comments posted by users with the substring ```bot''' or ```mod''' in their name and users with over 2000 comments in the dataset.
     \item We filter out low-information comments that are shorter than 250 tokens in length.
     \item Finally, we retain users with at least $100$ comments  for the user inference task, leading to around $5K$ users.
 \end{itemize}

\myparagraph{Reddit Small}
We also create a smaller version of this dataset with 4 months' data (the rest of the preprocessing pipeline remains the same). This gives us a dataset which is roughly half the size of the original one after filtering --- we denote this as ``{Reddit Comments (Small)}'' in \Cref{tab:datasets-arxiv}.

Although the unprocessed version of the small 4-month dataset is a subset of the unprocessed 6-month dataset, this is not longer the case after processing. After processing, 2626 users of the original 2774 users in the 4 month dataset were retained in the 6 month dataset. The other 148 users went over the 2000 comment threshold due to the additional 2 months of data and were filtered out as a part of the bot-filtering heuristic.
Note also that the held-in and held-out split between the two Reddit datasets is different (of the 1324 users in the 4-month training set, only 618 are in the 6-month training set). Still, we believe that a comparison between these two datasets gives a reasonable approximation how user inference changes with the scale of the dataset due to the larger number of users.
These results are given in \Cref{sec:a:results:data-size}.

\item \textbf{CC News}\footnote{\url{https://huggingface.co/datasets/cc_news}} \cite{Hamborg2017,charles2023towards}: Each example is a news article published on the Internet between January 2017 and December 2019. We define a user associated with an article to be the web domain where the article was found (e.g., nytimes.com). While CC News is not user-generated data (such as emails or posts used for the other datasets), it is a large group-partitioned dataset and has been used as a public benchmark for user-stratified federated learning applications~\cite{charles2023towards}. We note that this practice is common with other group-partitioned web datasets such as Stack Overflow~\cite{reddi2021adaptive}.


\item \textbf{Enron Emails}\footnote{\url{https://www.cs.cmu.edu/~enron/}} \cite{Klimt2004IntroducingTE}: Each example is an email found in the account of employees of the Enron corporation prior to its collapse. We define the user associated with an email to be the email address that sent an email. 

The original dataset contains a dump of emails in various folders of each user, e.g., ``inbox'', ``sent'', ``calendar'', ``notes'', ``deleted items'', etc. Thus, it contains a set of emails sent and received by each user. In some cases, each user also has multiple email addresses. Thus we take the following preprocessing steps for each user:
\begin{itemize}
    \item We list all the candidate sender's email address values on emails for a given user.
    \item We filter and keep candidate email addresses that contain the last name of the user, as inferred from the user name (assuming the user name is \texttt{<last name>-<first initial>}), also appears in the email.\footnote{
        This processing omits some users. For instance, the most frequently appearing sender's email of the user ``crandell-s'' with inferred last name ``crandell'' is \texttt{sean.crandall@enron.com}. It is thus omitted by the preprocessing.
    }
    \item We associate the most frequently appearing sender's email address from the remaining candidates.
    \item Finally, this dataset contains duplicates (e.g. the same email appears in the ``inbox'' and ``calendar'' folders). We then explicitly deduplicate all emails sent by this email address to remove exact duplicates. This gives the final set of examples for each user.
\end{itemize}
We verified that each of the remaining 138 users had their unique email addresses.

\item \textbf{ArXiv Abstracts}\footnote{\url{https://huggingface.co/datasets/gfissore/arxiv-abstracts-2021}} \cite{clement2019arxiv}: Each example is a scientific abstract posted to the ArXiv pre-print server through the end of 2021. We define the user associated with an abstract to be the first author of the paper. Note that this notion of author may not always reflect who actually wrote the abstract in case of collaborative papers. As we do not have access to perfect ground truth in this case, there is a possibility that the user labeling might have some errors (e.g. a non-first author wrote an abstract or multiple users collaborated on the same abstract). Thus, we postpone the results for the ArXiv Abstracts dataset to \Cref{appendix:expt-results}. 
See \Cref{tab:datasets-arxiv} for statistics of the ArXiv dataset.

\end{itemize}

\begin{table*}[t!]
\centering
\renewcommand{\arraystretch}{1.2}
\small
\begin{tabular}{llrrrrrrr} 
\toprule
\multirow{2}{*}{\textbf{Dataset}} & 
\multirow{2}{*}{\textbf{User Field}} & 
\multirow{2}{*}{\textbf{\#Users}} &
\multirow{2}{*}{\textbf{\#Examples}} &
\multicolumn{5}{c}{\textbf{Percentiles of Examples/User}} \\
\cmidrule{5-9}
& & & & \multicolumn{1}{l}{$\mathbf{P_{0}}$} & \multicolumn{1}{l}{$\mathbf{P_{25}}$} & \multicolumn{1}{l}{$\mathbf{P_{50}}$} & \multicolumn{1}{l}{$\mathbf{P_{75}}$} & \multicolumn{1}{l}{$\mathbf{P_{100}}$}
\\ \hline
ArXiv Abstracts    & 
Submitter              &
$16511$      &
$625K$  &  
$20$      &
$24$      &
$30$      &
$41$      &
$3204$ \\
Reddit Comments (Small) &
User Name &
2774 &
$537K$ &
$100$ &
$115$ &
$141$ &
$194$ &
$1662$
\\ \bottomrule
\end{tabular}
\caption{\small Summary statistics for additional datasets.} \label{tab:datasets-arxiv}
\end{table*}

Despite the imperfect ground truth labeling of the ArXiv datasets, we believe that evaluating the proposed user inference attack reveals the risk of privacy leakage in fine-tuned LLMs for two reasons.
First, the fact that we have significant privacy leakage despite imperfect user labeling suggests that the attack will only get stronger if we had perfect ground truth user labeling and non-overlapping users. This is because mixing distributions only brings them closer, as shown in \Cref{prop:mixing} below. Second, our experiments on canary users are not impacted at all by the possible overlap in user labeling, since we create our own synthetically-generated canaries to evaluate worst-case privacy leakage. 

\begin{proposition}[Mixing Distributions Brings Them Closer]
\label{prop:mixing}
Let $P, Q$ be two user distributions over text. Suppose mislabeling leads to the respective mixture distributions of $P' = \lambda P + (1 - \lambda)Q $ and $Q' = \mu Q + (1 - \mu)P$ for some $\lambda, \mu \in [0, 1]$. Then, we have, $\kl(P' \Vert Q') \le \kl(P \Vert Q)$.
\end{proposition}
\begin{proof}
The proof follows from the convexity of the KL divergence in both its arguments. Indeed, we have,
\[
\kl(P \Vert \mu Q + (1-\mu)P)
\le \mu \, \kl(P \vert Q) + (1-\mu) \,  \kl(P \Vert P) \le \kl(P \vert Q) \,,
\]
since $0\le \mu \le 1$ and $\kl(P \Vert P) = 0$. A similar reasoning for the first argument of the KL divergence completes the proof.
\end{proof}

\myparagraph{Preprocessing}
Before fine-tuning models on these datasets we perform the following preprocessing steps to make them suitable for evaluating user inference.
\begin{enumerate}
    \item We filter out users with fewer than a minimum number of samples ($20$, $100$, $30$, and $150$ samples for ArXiv, Reddit, CC News, and Enron respectively). These thresholds were selected prior to any experiments to balance the following considerations: (1) each user must have enough data to provide the attacker with enough samples to make user inference feasible and (2) the filtering should not remove so many users that the fine-tuning dataset becomes too small. The summary statistics of each dataset after filtering are shown in Table~\ref{tab:datasets}.
    \item We reserve $10\%$ of the data for validation and test sets
    \item We split the remaining $90\%$ of samples into a held-in set and held-out set, each containing half of the users. The held-in set is used for fine-tuning models and the held-out set is used for attack evaluation.
    \item For each user in the held-in and held-out sets, we reserve $10\%$ of the samples as the attacker's knowledge about each user. These samples are never used for fine-tuning. 

\end{enumerate}

\myparagraph{Target Models}
We evaluate user inference attacks on the $125$M and $1.3$B parameter models from the GPT-Neo~\cite{gpt-neo} model suite.
For each experiment, we fine-tune all parameters of these models for $10$ epochs. We use the the Adam optimizer \cite{kingma2017adam} with a learning rate of $5\times10^{-5}$, a linearly decaying learning rate schedule with a warmup period of $200$ steps, and a batch size of $8$. After training, we select the checkpoint achieving the minimum loss on validation data from the users held in to training, and use this checkpoint to evaluate user inference attacks.

We train models on servers with one NVIDIA A100 GPU and $256$ GB of memory. Each fine-tuning run took approximately $16$ hours to complete for GPT-Neo $125$M and $100$ hours for GPT-Neo $1.3$B.

\myparagraph{Attack Evaluation}
We evaluate attacks by computing the attack statistic from Section~\ref{sec:attack} for each held-in user that contributed data to the fine-tuning dataset, as well as the remaining held-out set of users. With these user-level statistics, we compute a Receiver Operating Characteristic (ROC) curve and report the area under this curve (AUROC) as our metric of attack performance. This metric has been used recently to evaluate the performance of membership inference attacks \citet{LiRA}, and it provides a full spectrum of the attack effectiveness (True Positive Rates at fixed False Positive Rates). By reporting the AUROC, we do not need to select a threshold $\tau$ for our attack statistic, but rather we report the aggregate performance of the attack across all possible thresholds. 

\subsection{Canary User Construction}
\label{sec:a:expt:canary}

We evaluate worst-case risk of user inference by injecting synthetic canary users into the fine-tuning data from CC News, ArXiv Abstracts, and Reddit Comments. These canaries were constructed by taking real users and replicating a shared substring in all of that user's examples. This construction is meant to create canary users that are both realistic (i.e. not substantially outlying compared to the true user population) but also easy to perform user inference on. The algorithm used to construct canaries is shown in \Cref{alg:canary}.

\begin{algorithm}
\caption{Synthetic canary user construction}\label{alg:canary}
\begin{algorithmic}
\Require Substring lengths $L = [l_1, \dots l_n]$,
canaries per substring length $N$, set of real users $U_R$
\Ensure Set of canary users $U_C$
\State $U_C \gets \emptyset$
\For{$l$ in $L$}
    \For{$i$ up to $N$}
        \State Uniformly sample user $u$ from $U_R$
        \State Uniformly sample example $x$ from $u$'s data
        \State Uniformly sample $l$-token substring $s$ from $x$
        \State $u_c \gets \emptyset$ \Comment{Initialize canary user with no data}
        \For{$x$ in $u$}
            \State $x_c \gets \text{InsertSubstringAtRandomLocation}(x, s)$
            \State Add example $x_c$ to user $u_c$
        \EndFor
        \State Add user $u_c$ to $U_C$
        \State Remove user $u$ from $U_R$

    \EndFor
\EndFor
\end{algorithmic}
\end{algorithm}

\subsection{Mitigation Definitions}
\label{sec:a:expt:mitigation}

In \Cref{sec:results} we explore heuristics for mitigating privacy attacks.
We give precise definitions of the batch and per-example gradient clipping.

Batch gradient clipping restricts the norm of a single batch gradient to be at most $C$:
\begin{align*}
    \hat{g}_t = \frac{\min(C, \lVert \nabla_{\theta_t} l(\xv) \rVert)}{\lVert \nabla_{\theta_t} l(\xv) \rVert}  \nabla_{\theta_t} l(\xv) \,.
\end{align*}

Per-example gradient clipping restricts the norm of a single example's gradient to be at most $C$ before aggregating the gradients into a batch gradient: 
\begin{align*}
    \hat{g}_t = \sum_{i=1}^n \frac{\min(C, \lVert \nabla_{\theta_t} l(\xv^{(i)}) \rVert)}{\lVert \nabla_{\theta_t} l(\xv^{(i)}) \rVert}   \nabla_{\theta_t} l(\xv^{(i)}) \,.
\end{align*}

The batch or per-example clipped gradient $\hat{g}_t$, is then passed to the optimizer as if it were the true gradient.

For all experiments involving gradient clipping, we selected the clipping norm, $C$, by recording the gradient norms during a standard training run and setting $C$ to the minimum gradient norm. In practice this resulted in clipping nearly all batch/per-example gradients during training.

\subsection{Example-Level Differential Privacy: Hyperparameter Tuning}
\label{appendix:dp-tuning}

We now describe the hyperparameter tuning strategy for the example-level DP experiments reported in \Cref{tab:enron_dp}.
Broadly, we follow the guidelines outlined by \citet{ponomareva2023dp}.
Specifically, the tuning procedure is as follows:

\begin{itemize}[nosep]
    \item The Enron dataset has $n=41000$ examples from held-in users used for training. The Non-private training of reaches its best validation loss in about $3$ epochs or $T=15K$ steps. We keep this fixed for the batch size tuning.
    \item \textbf{Tuning the batch size}: For each privacy budget $\eps$ and batch size $b$, we obtain the noise multiplier $\sigma$ such that the private sum $\sum_{i=1}^b g_i + \mathcal{N}(0, \sigma^2)$ repeated $T$ times (one for each step of training) is $(\eps, \delta)$-DP, assuming that each $\|g_i\|_2 \le 1$. The noise scale per average gradient is then $\sigma / \sqrt{b}$. This is the \textbf{inverse signal-to-noise ratio} and is plotted in \Cref{fig:dp:snr}. 
    
    We fix a batch size of $1024$ as the curves flatten out by this point for all the values of $\eps$ considered. See also \cite[Fig. 1]{ponomareva2023dp}.

    \item \textbf{Tuning the number of steps}: Now that we fixed the batch size, we train for as many steps as possible in a 24 hour time limit (this is $12\times$ more expensive than non-private training). Note that DP training is slower due to the need to calculate per-example gradients. This turns out to be around 50 epochs or 1200 steps.
    
    \item \textbf{Tuning the learning rate}:
    We tune the learning rate while keeping the gradient clipping norm at $C=1.0$ (note that non-private training is not sensitive to the value of gradient clip norm).
    We experiment with different learning rate and pick $3 \times 10^{-4}$ as it has the best validation loss for $\eps=8$ (see \Cref{fig:dp:lr}). We use this learning rate for all values of $\eps$.
\end{itemize}

\begin{figure*}[t]
    \begin{subfigure}[b]{0.45\linewidth}
\centering
    \includegraphics[width=\linewidth]{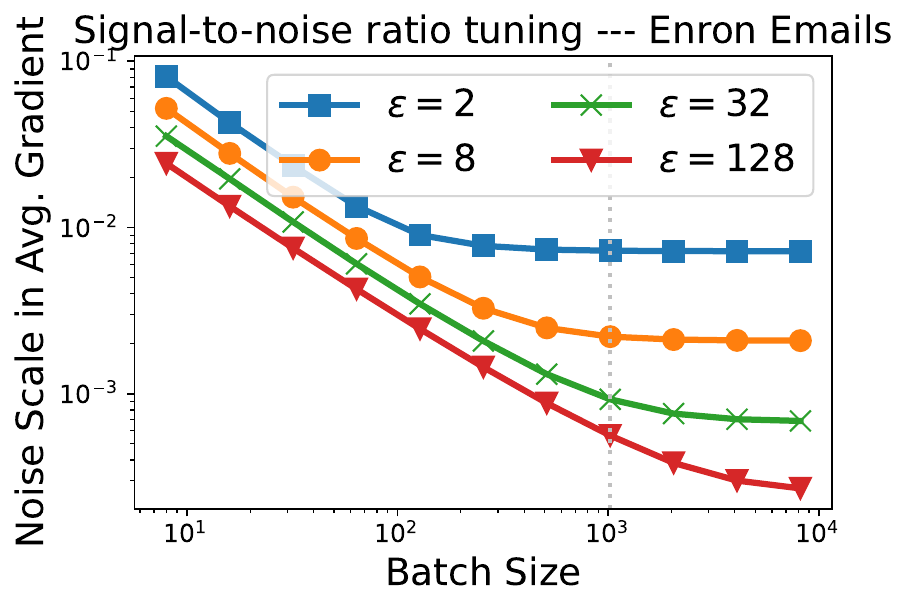}
    \caption{\small The scale of the noise added to the average gradients.}
    \label{fig:dp:snr}
    \end{subfigure}
    \begin{subfigure}[b]{0.45\linewidth}
\centering
    \includegraphics[width=\linewidth]{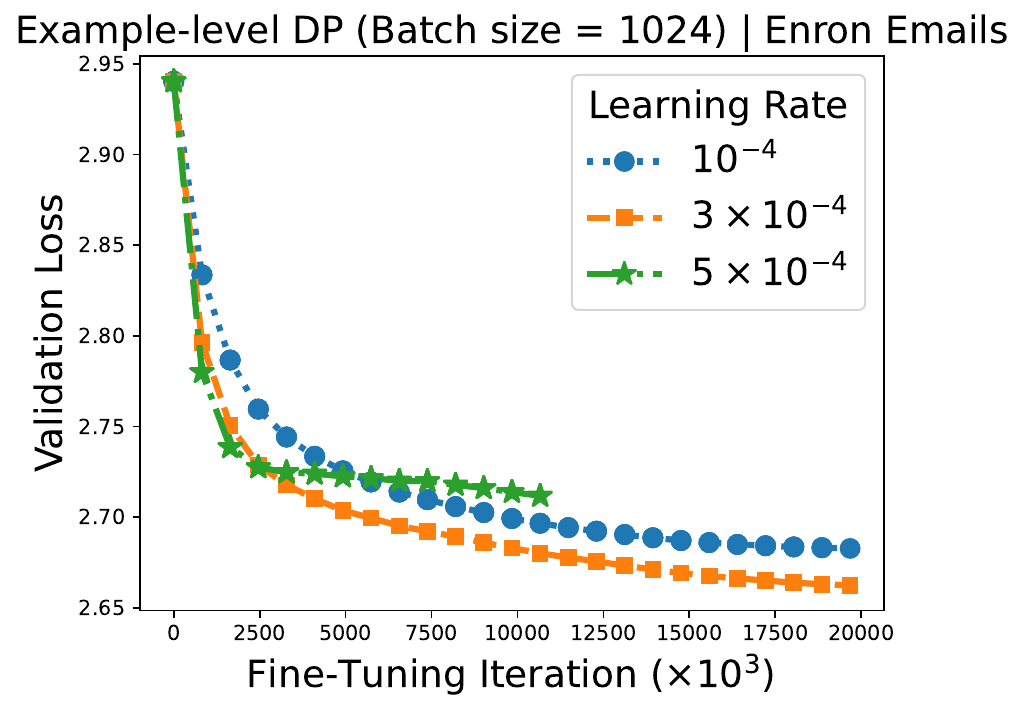}
    \caption{\small Tuning the learning rate with $\eps=8$.}
    \label{fig:dp:lr}
\end{subfigure}
\caption{\small Tuning the parameters for example-level DP on the Enron dataset.}
\end{figure*}

\subsection{Analysis of Duplicates in CC News}
\label{sec:ccnews-dup}

The CC News dataset from HuggingFace Datasets has $708241$ examples, each of which has the following fields: web domain (i.e., the ``user''), the text (i.e. the body of the article), the date of publishing, the article title, and the URL. Each example has a \emph{unique URL}.
However, the text of the articles from a given domain are not all unique. In fact, there only $628801$ articles (i.e., $88.8\%$ of the original dataset) after removing exact text duplicates from a given domain. While all of the duplicates have unique URLs, $43$K out of the identified $80$K duplicates have unique article titles).

We list some examples of exact duplicates below:
\begin{itemize}
    \item \texttt{which.co.uk}: ``We always recommend that before selecting or making any important decisions about a care home you take the time to check that it is right for your or your relative's particular circumstances. Any description and indication of services and facilities on this page have been provided to us by the relevant care home and we cannot take any responsibility for any errors or other inaccuracies. However, please email us on the address you will find on our About us page if you think any of the information on this page is missing and / or incorrect.'' has $3$K duplicates.
    \item \texttt{amarujala.com}: ``Read the latest and breaking Hindi news on amarujala.com. Get live Hindi news about India and the World from politics, sports, bollywood, business, cities, lifestyle, astrology, spirituality, jobs and much more. Register with amarujala.com to get all the latest Hindi news updates as they happen.'' has $2.2$K duplicates.
    \item \texttt{saucey.com}: ``Thank you for submitting a review! Your input is very much appreciated. Share it with your friends so they can enjoy it too!'' has $1$K duplicates.
    \item \texttt{fox.com}: ``Get the new app. Now including FX, National Geographic, and hundreds of movies on all your devices.'' has $0.6$K duplicates.
    \item \texttt{slideshare.net}: ``We use your LinkedIn profile and activity data to personalize ads and to show you more relevant ads. You can change your ad preferences anytime.'' has $0.5$K duplicates.
    \item \texttt{ft.com}: ``\$11.77 per week * Purchase a Newspaper + Premium Digital subscription for \$11.77 per week. You will be billed \$66.30 per month after the trial ends'' has $200$ duplicates.
    \item \texttt{uk.reuters.com}: ``Bank of America to lay off more workers (June 15): Bank of America Corp has begun laying off employees in its operations and technology division, part of the second-largest U.S. bank's plan to cut costs.'' has $52$ copies.
\end{itemize}

As shown in \Cref{fig:ccnews-duplicates}, a small fraction of examples account for a large number of duplicates (the right end of the plot). Most of such examples are typically web scraping errors. Some of the web domains have legitimate news article repetitions, such as the last example above.
In general, these experiments suggest that exact or approximate deduplication for the data contributed by each deduplication is a low cost preprocessing step that can moderately reduce the privacy risks posed by user inference.

\begin{figure*}
    \centering
    \includegraphics[width=0.4\linewidth]{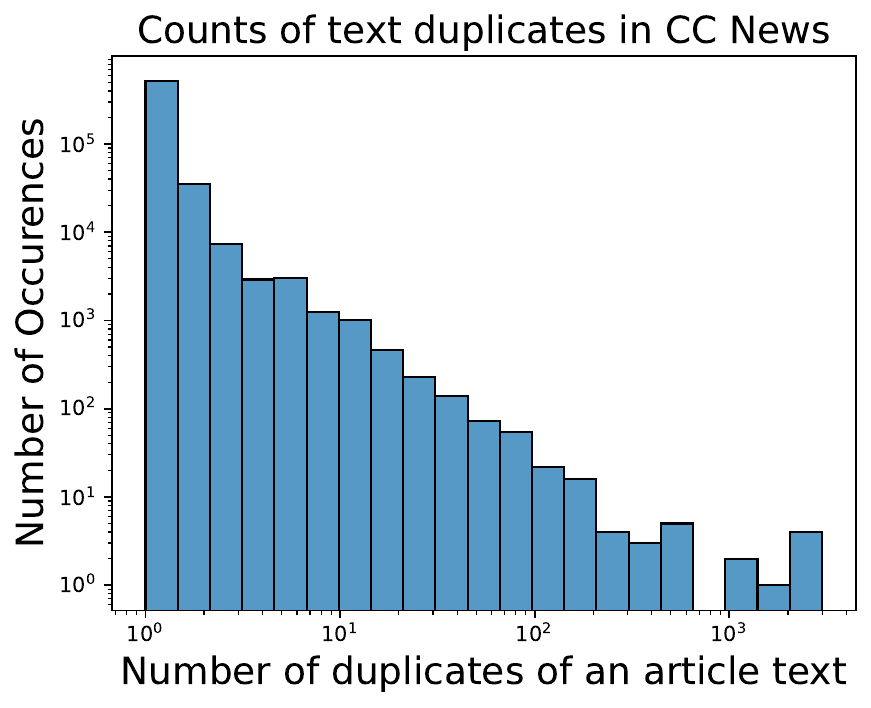}
    \caption{\small Histogram of number of duplicates in CC News. The right side of the plot shows a small number of unique articles have a large number of repetitions.}
    \label{fig:ccnews-duplicates}
\end{figure*}
\section{Additional Experimental Results} 
\label{appendix:expt-results}

We give full results on the ArXiv Abstracts dataset, provide further results for example-level DP,  and
run additional ablations.
Specifically, the outline of the section is:
\begin{itemize}[nosep]
    \item \Cref{sec:a:results:Arxiv}:
    Additional experimental results showing user inference on the ArXiv dataset.
    \item \Cref{sec:a:results:data-size}: Additional experiments on the effect of increasing the dataset size.
    \item \Cref{sec:a:results:tpr-stats}: Tables of TPR statistics at particular values of small FPR.
    \item \Cref{sec:a:results:dp-roc}: ROC curves corresponding to the example-level DP experiment (\Cref{tab:enron_dp}).
    \item \Cref{sec:a:results:ablations}: Additional ablations on the aggregation function and reference model.
\end{itemize}

\begin{figure*}
\begin{subfigure}[b]{\linewidth}
\centering
    \includegraphics[width=0.7\linewidth]{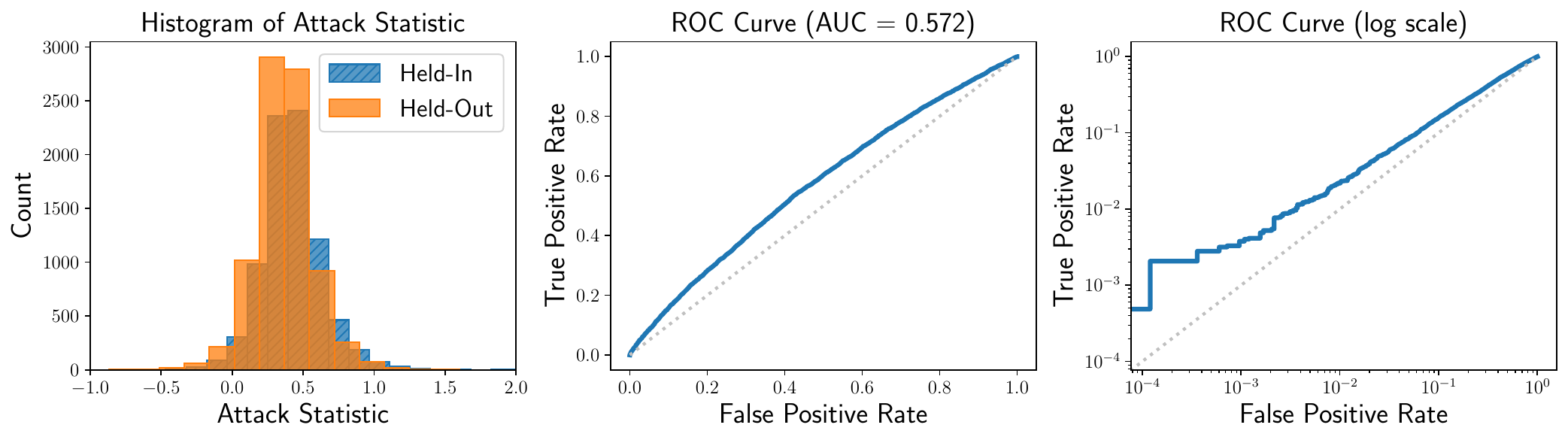}
    \caption{Main attack results (cf. \Cref{fig:hists_aurocs}): histograms of test statistics for held-in and held-out users and ROC curve.}
    \label{fig:arxiv:main}
\end{subfigure}

\begin{subfigure}[b]{0.48\linewidth}
\centering
    \includegraphics[width=\linewidth]{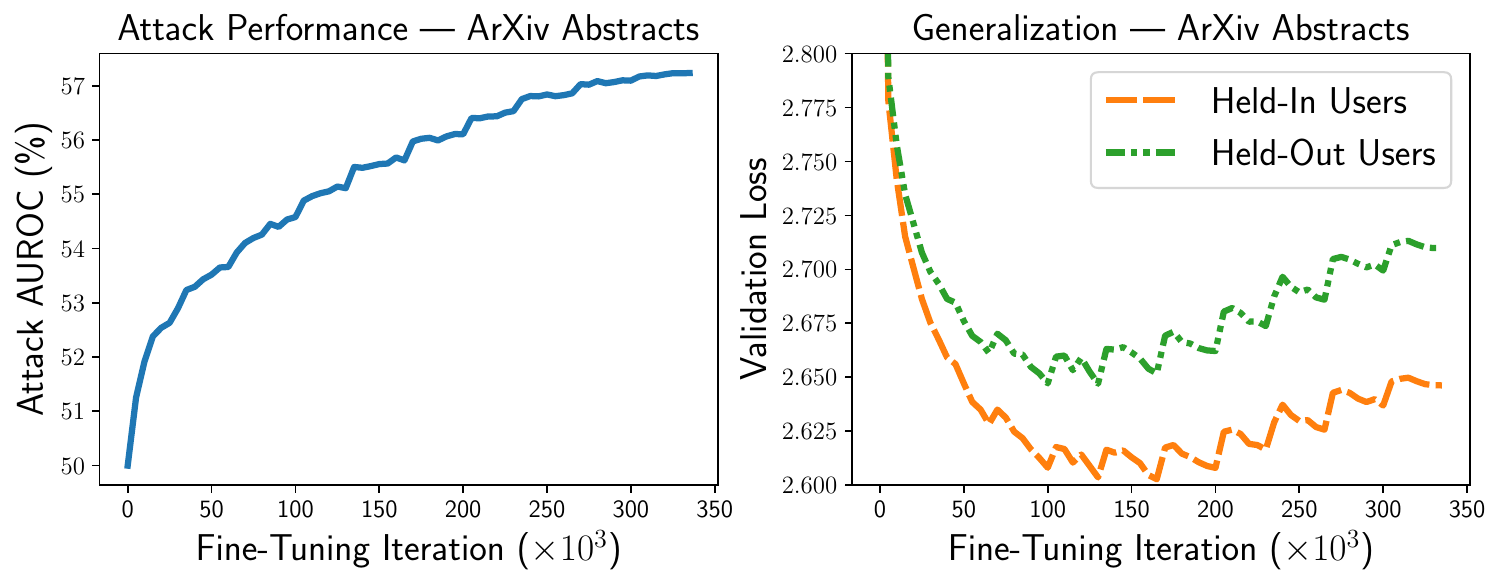}
    \caption{Attack results over the course of training (cf. \Cref{fig:training_run}).}
    \label{fig:arxiv:training}
\end{subfigure}
\begin{subfigure}[b]{0.48\linewidth}
\centering
    \includegraphics[width=\linewidth]{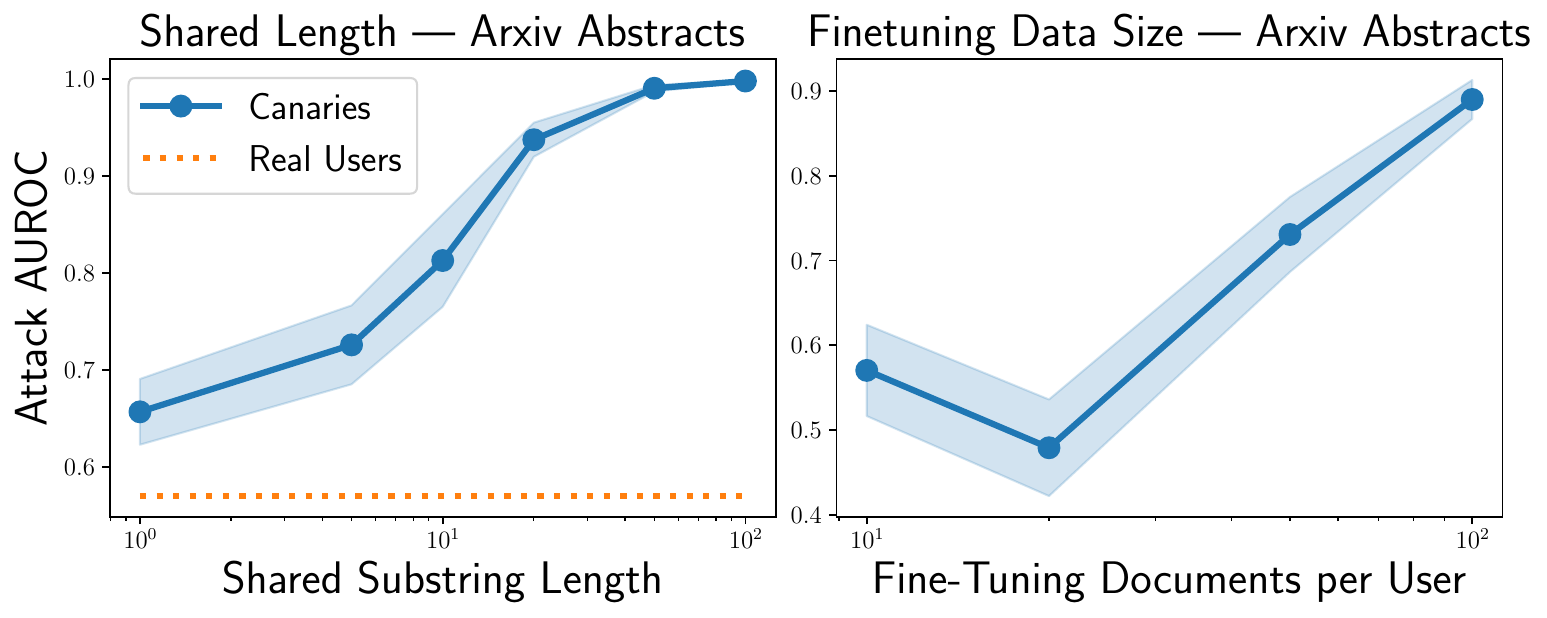}
    \caption{Attack results with canaries (cf. \Cref{fig:canaries}).}
    \label{fig:arxiv:canaries}
\end{subfigure}

\caption{Results on the ArXiv Abstracts dataset.}
\label{fig:arxiv}
\end{figure*}

\subsection{Results on the ArXiv Abstracts Dataset}
\label{sec:a:results:Arxiv}

\Cref{fig:arxiv} shows the results for the ArXiv Abstracts dataset. Broadly, we find that the results are qualitatively similar to those of Reddit Comments and CC News.

Quantitatively, the attack AUROC is $57\%$, in between Reddit ($56\%$) and CC News ($66\%$).
\Cref{fig:arxiv:training} shows the user-level generalization and attack performance for the ArXiv dataset.
The Spearman rank correlation between the user-level generalization gap and the attack AUROC is at least $99.8\%$, which is higher than the $99.4\%$ of CC News (although the trend is not as clear visually).
This reiterates the close relation between user-level overfitting and user inference.
Finally, the results of \Cref{fig:arxiv:canaries} are also nearly identical to those of \Cref{fig:canaries}, reiterating their conclusions.

\subsection{Effect of Increasing the Dataset Size: Reddit}
\label{sec:a:results:data-size}

We now compare the effect increasing the size of the dataset has on user inference. To be precise, we compare the full Reddit dataset that contains 6 months of scraped comments with a smaller version that uses 4 months of data (see \Cref{sec:a:expt:ds} and \Cref{fig:reddit_data_size} for details).

We find in \Cref{fig:reddit_data_size:roc} that increasing the size of the dataset leads to a uniformly smaller ROC curve, including a reduction in AUROC ($60\%$ to $56\%$) and a smaller TPR at various FPR values.

\begin{figure}[bth!]
\centering
\begin{subfigure}[b]{0.465\linewidth}
\centering
\includegraphics[width=\linewidth]{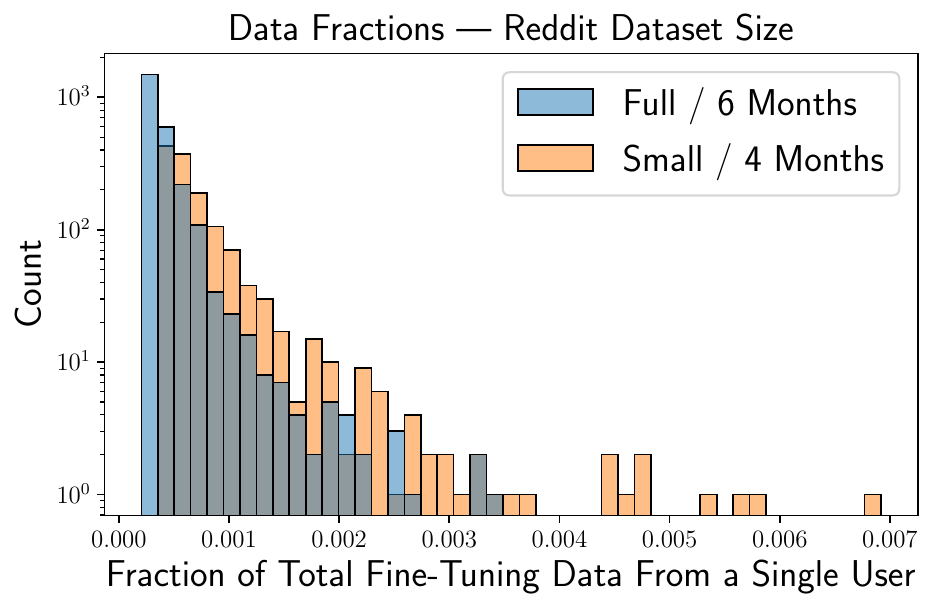}
\caption{\small
    Histogram of fraction of data per user.
} \label{fig:reddit_data_size}
\end{subfigure}
\begin{subfigure}[b]{0.45\linewidth}
\centering
\includegraphics[width=\linewidth]{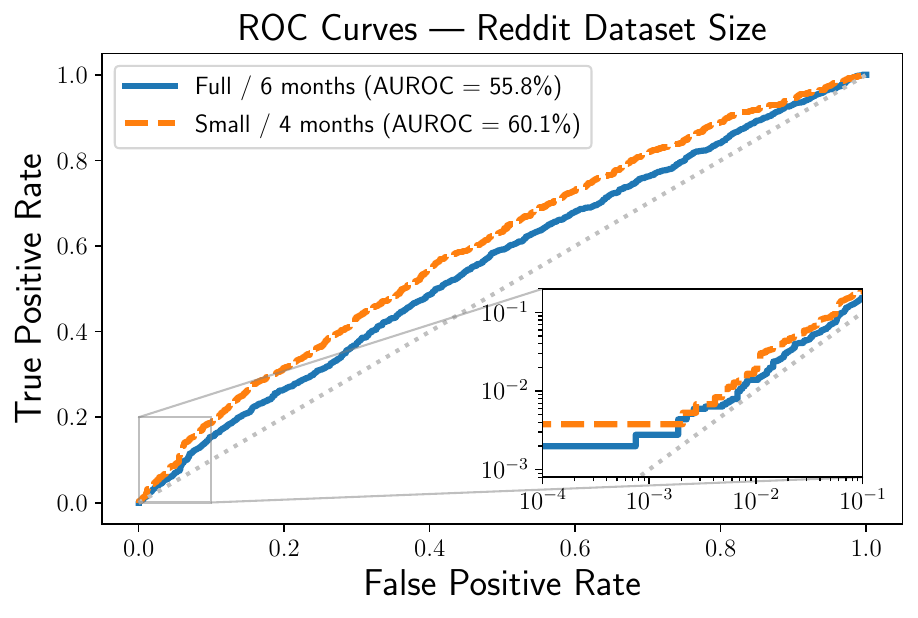}
\caption{\small
The corresponding ROC curves.
}
\label{fig:reddit_data_size:roc}
\end{subfigure}
\caption{Effect of increasing the fraction of data contributed by each user: Since Reddit Full (6 Months) contains more users than Reddit Small (4 Months), each user contributes a smaller fraction of the total fine-tuning dataset. As a result, the user inference attack on Reddit Full is less successful, which agrees with the intuition from \Cref{thm:analysis}.}
\end{figure}

\subsection{Attack TPR at low FPR}
\label{sec:a:results:tpr-stats}

We give some numerical values of the attack TPR and specific low FPR values.

\myparagraph{Main experiment}
While \Cref{fig:hists_aurocs} summarizes the attack performance with the AUROC, we give the attack TPR at particular FPR values in \Cref{tab:tpr_at_small_fpr}.
This result shows that while Enron's AUROC is large, its TPR at FPR$=1\%$ at $4.41\%$ is comparable to the $4.41\%$ of CC News.  However, for FPR$=5\%$, the TPR for Enron jumps to nearly $28\%$, which is much larger than the $11\%$ of CC News.

\begin{table*}
\small
\centering
\renewcommand{\arraystretch}{1.2}
\begin{tabular}{lrrrr}
\toprule
\textbf{FPR \%} &  
\multicolumn{4}{c}{\textbf{TPR\%}}
\\
\cmidrule{2-5}
& \textbf{Reddit} &  \textbf{CC News} &  \textbf{Enron} &  \textbf{ArXiv} \\
\midrule
$0.1$  &   $0.28$ &   $1.18$ &   N/A &   $0.38$ \\
$0.5$  &   $0.67$ &   $2.76$ &   N/A &   $1.31$ \\
$1$  &   $1.47$ &   $4.33$ &   $4.41$ &   $2.24$ \\
$5$  &   $7.05$ &  $11.02$ &  $27.94$ &   $8.44$ \\
$10$ &  $15.45$ &  $18.27$ &  $57.35$ &  $15.77$ \\
\bottomrule
\end{tabular}
\caption{\small Attack TPR at small FPR values corresponding to \Cref{fig:hists_aurocs}.}
\label{tab:tpr_at_small_fpr}
\end{table*}

\myparagraph{CC News Deduplication}
The TPR statistics at low FPR are given in \Cref{tab:ccnews-dedup}.

\begin{table*}
\centering
\renewcommand{\arraystretch}{1.2}
\begin{tabular}{lrrrrrr}
\toprule
\textbf{CC News Variant} &  \textbf{AUROC \%} & 
\multicolumn{5}{c}{\textbf{TPR\% at FPR $=$ }} \\
\cmidrule{3-7}
& & $0.1\%$ &    \textbf{$0.5\%$} &      \textbf{$1\%$} &       \textbf{$5\%$} &      \textbf{$10\%$} \\
\midrule
\textbf{Original} &  65.73 &  1.18 &  2.76 &  4.33 &  11.02 &  18.27 \\
\textbf{Deduplicated}    &  59.08 &  0.58 &  1.00 &  1.75 &   7.32 &  11.31 \\
\bottomrule
\end{tabular}
\caption{\small Effect of within-user deduplication: Attack TPR at small FPR values corresponding to \Cref{fig:dedup}.}
\label{tab:ccnews-dedup}
\end{table*}

\subsection{ROC Curves for Example-Level Differential Privacy}
\label{sec:a:results:dp-roc}

The ROC curves corresponding to the example-level differential privacy is given in \Cref{fig:dp-roc}. The ROC curves reveal that while example-level differential privacy (DP) reduces the attack AUROC, we find that the TPR at low FPR remains unchanged. In particular, for FPR $=3\%$, we have TPR $=6\%$ for the non-private version but TPR $=10\%$ for $\varepsilon = 32$. This shows that example-level DP is ineffective at fully thwarting the risk of user inference.

\begin{figure*}
\centering
    \includegraphics[width=0.75\linewidth]{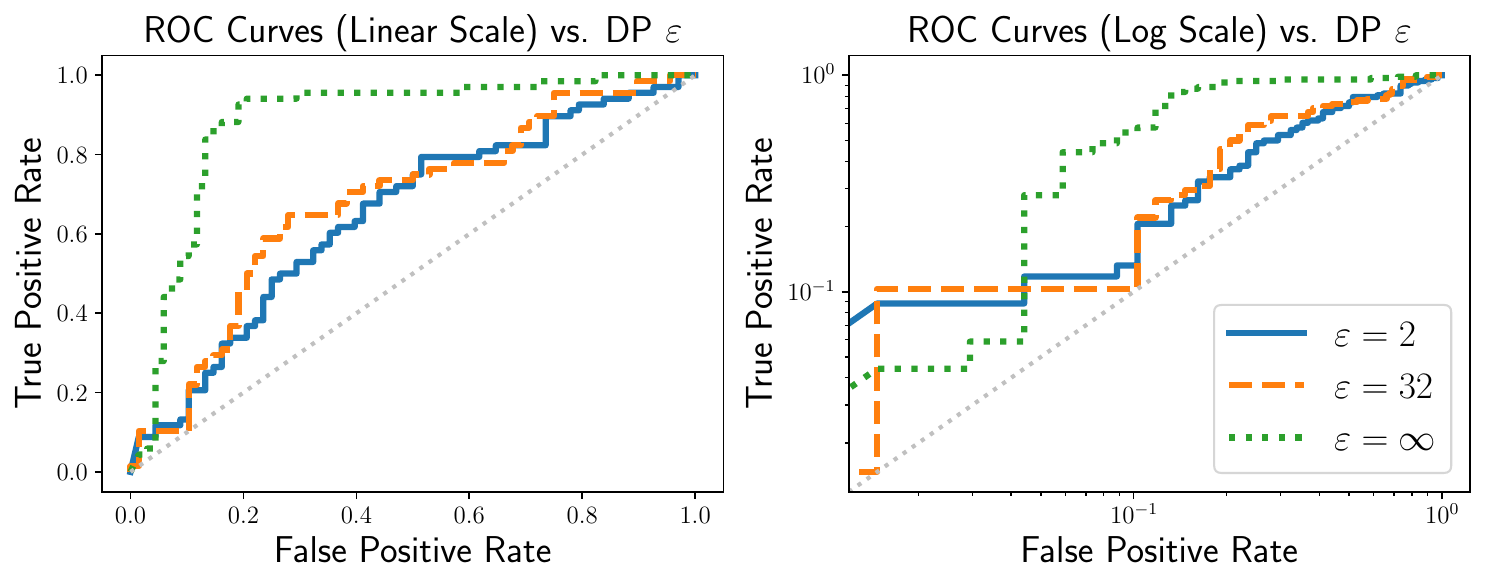}
    \caption{ROC curves (linear and log scale) for the example-level differential privacy on the Enron Emails dataset.}
    \label{fig:dp-roc}
\end{figure*}

\subsection{Additional Ablations} 
\label{sec:a:results:ablations}

The user inference attacks implemented in the main paper use the pre-trained LLM as a reference model and compute the attack statistic as a mean of log-likelihood ratios described in Section~\ref{sec:attack}. In this section, we study different choices of reference model and different methods of aggregating example-level log-likelihood ratios. For each of the attack evaluation datasets, we test different choices of reference model and aggregation function for performing user inference on a fine-tuned GPT-Neo $125$M model.

In \Cref{tab:aggregation} we test three methods of aggregating example-level statistics and find that averaging the log-likelihood ratio outperforms using the minimum or maximum per-example ratio. Additionally, in \Cref{tab:reference_model} we find that using the pre-trained GPT-Neo model as the reference model outperforms using an independently trained model of equivalent size, such as OPT~\cite{zhang2022opt} or GPT-2~\cite{Radford2019LanguageMA}. However, in the case that an attacker does not know or have access to the pre-trained model, using an independently trained LLM as a reference still yields strong attack performance.

\begin{table}[H]
\centering
\renewcommand{\arraystretch}{1.2}
\small
\begin{tabular}{crrrr}
\toprule
\begin{tabular}{c}\textbf{Attack Statistic} \\ \textbf{Aggregation} \end{tabular}        & 
\textbf{Reddit Comments} &
\textbf{ArXiv Abstracts} & 
\textbf{CC News} & 
\textbf{Enron Emails} 
\\ \hline
Mean & 
$\tabemph{} \mathbf{56.0 \pm 0.7}$
    &
$\tabemph{} \mathbf{57.2 \pm 0.4}$       & 
$\tabemph{} \mathbf{65.7 \pm 1.1}$      & 
$\tabemph{} \mathbf{87.3 \pm 3.3}$                      
\\
Max          & 
$54.5 \pm 0.8$ & 
$56.7 \pm 0.4$      & 
$62.1 \pm 1.1$               & 
$71.1 \pm 4.0$               
\\
Min          & 
$54.6 \pm 0.8$ & 
$55.3 \pm 0.4$               & 
$63.3 \pm 1.0$                & 
$57.9 \pm 4.0$     
\\                    
\bottomrule
\end{tabular}
\caption{\small 
\textbf{Attack statistic design}:
We compare the default mean aggregation of per-document statistics $\log(p_\theta(\xv\pow{i}) / \pRef(\xv\pow{i}))$ in the attack statistic (\Cref{sec:attack}) with the min/max over documents $i=1, \ldots, m$.
We show the mean and std AUROC over 100 bootstrap samples of the held-in and held-out users.
} \label{tab:aggregation}
\end{table}
\begin{table}[H]
\centering
\renewcommand{\arraystretch}{1.2}
\small
\begin{tabular}{lrrr} 
\toprule
\textbf{Reference Model} & 
\textbf{ArXiv Abstracts}  & 
\textbf{CC News} & 
\textbf{Enron Emails}
\\ \hline
GPT-Neo 125M$^*$    & 
$\tabemph{} \mathbf{57.2 \pm 0.4}$              &
$\tabemph{}\mathbf{65.8 \pm 1.1}$              &
$\tabemph{}\mathbf{87.8 \pm 3.5}$                      
\\
GPT-2 124M    & 
$53.1 \pm 0.5$                             & 
$\tabemph{}\mathbf{65.7 \pm 1.2}$                       & 
$74.1 \pm 4.5$                                
\\
OPT 125M        & 
$53.7 \pm 0.5$                            & 
$62.0 \pm 1.2$                        & 
$77.9 \pm 4.2$  
\\ \bottomrule
\end{tabular}
\caption{\small \textbf{Effect of the reference model}: We show the user inference attack AUROC $(\%)$ for different choices of the reference model $\pRef$, including the pretrained model $p_{\theta_0}$ (GPT-Neo 125M, denoted by $^*$).
We show the mean and std AUROC over 100 bootstrap samples of the held-in and held-out users.
} \label{tab:reference_model}
\end{table}

\section{Discussion on User-Level DP}
\label{appendix:user-level-dp}

Differential privacy (DP) at the user-level gives quantitative and provable guarantees that the presence or absence of \emph{one user's data} is indistinguishable. Concretely, a training procedure is \textbf{$(\eps, \delta)$-DP at the user level} if the model $p_{\theta}$ trained on the data from set $U$ of users and a model $p_{\theta, u}$ trained on data from users $U \cup \{u\}$ satisfies
\begin{align} \label{eq:user-level-dp}
 \prob(p_\theta \in A) \le \exp(\eps) \, \prob(p_{\theta, u} \in A) + \delta \,,
\end{align}
and analogously with $p_\theta$, $p_{\theta, u}$ interchanged, for any outcome set $A$ of models, any user $u$ and any $U$ of users.
Here, $\eps$ is known as the privacy budget and a smaller value of $\eps$ denotes greater privacy.

In practice, this involves ``clipping'' the user-level contribution and adding noise calibrated to the privacy level~\cite{mcmahan2017learning}.

\myparagraph{The promise of user-level DP}
User-level DP is the strongest form of protection against user inference. For instance, suppose we take
\[
    A = \left\{\theta \, :\, \frac{1}{m} \sum_{i=1}^m \log\left( \frac{p_\theta(\xv\pow{i})}{\pRef(\xv\pow{i})} \right) \le \tau \right\}
\]
to be set of all models whose test statistic calculated on $\xv\pow{1:m} \sim \Puser_u^m$ is at most some threshold $\tau$. Then, the user-level DP guarantee \eqref{eq:user-level-dp} says that the test statistic between $p_{\theta}$ and $p_{\theta, u}$ are nearly indistinguishable (in the sense of \eqref{eq:user-level-dp}). In other words, the attack AUROC is provably bounded as function of the parameters $(\eps, \delta)$~\cite{kairouz2015composition}. 

User-level DP has successfully been deployed on industrial applications with user data~\cite{ramaswamy2020training,xu2023federated}. However, these applications are in the context of federated learning with small on-device models.

\myparagraph{The challenges of user-level DP}
While user-level DP is a natural solution to mitigate user inference, it involves several challenges, including fundamental dataset sizes, software/systems challenges, and a lack of understanding of empirical tradeoffs.

First, user-level DP can lead to a major drop in performance, especially if the number of users in the fine-tuning dataset is not very large. For instance, the Enron dataset with $O(150)$ users is definitely too small while CC news with $O(3000)$ users is still on the smaller side. It is common for studies on user-level DP to use datasets with $O(100K)$ users. For instance, the Stack Overflow dataset, previously used in the user-level DP literature, has around $350K$ users~\cite{kairouz2021practical}.

Second, user-aware training schemes including user-level DP and user-level clipping, require sophisticated user-sampling schemes. For instance, we may require operations of the form ``sample 4 users and return 2 samples from each''. On the software side, this requires fast per-user data loaders, which are not supported by standard training workflows, which are oblivious to the user-level structure in the data.

Third, user-level DP also requires careful accounting of user contributions per round and balancing user contributions per-round and the number of user participations over all rounds. The trade-offs involved here are not well-studied, and require a detailed investigation.

Finally, existing approaches require the datasets to be partitioned into disjoint user data subsets. Unfortunately, this is not always true in applications such as email threads (where multiple users contribute to the same thread) or collaborative documents. The ArXiv Abstracts dataset suffers from this latter issue as well. This is a promising direction for future work.

\myparagraph{Summary}
In summary, the experimental results we presented make a strong case for user-level DP at the LLM scale. Indeed, our results motivate the separate future research question on how to effectively apply user-level DP given accuracy and compute constraints.

\end{document}